\documentclass[12pt]{iopart}
\usepackage[T2A]{fontenc}
\usepackage{amsfonts}

 \expandafter\let\csname equation*\endcsname\relax
 \expandafter\let\csname endequation*\endcsname\relax

\usepackage{amsmath}
\usepackage{amsthm}
\usepackage{enumerate}
\usepackage{graphicx}

\usepackage{amsfonts}

\usepackage{epsfig}
\usepackage{graphicx}

\usepackage[usenames]{color}
\usepackage{colortbl}

\renewcommand{\e}{\mathrm{e}}
\renewcommand{\i}{\mathrm{i}}
\newcommand{\w}{\mathrm{w}}
\renewcommand{\d}{\mathrm{d}}
\newcommand{\ds}{\displaystyle}
\newcommand{\ord}{\mathrm{O}\left(}
\newcommand{\dsfrac}{\ds\frac}
\newcommand{\I}{\mathrm{I}}
\renewcommand{\(}{\left(}
\renewcommand{\)}{\right)}
\newcommand{\ol}{\overline}

\newtheorem*{prop*}{Proposition}
\newtheorem{teor}{Theorem}[section]
\newtheorem{lem}{Lemma}[section]

\newtheorem{com}{Comment}[section]
\newtheorem*{col*}{Collorary}

\renewcommand{\e}{\mathrm{e}}

\renewcommand{\i}{\mathrm{i}}
\renewcommand{\d}{\mathrm{d}}
\providecommand{\B}{\mathbf}

\renewcommand{\(}{\left(}
\renewcommand{\)}{\right)}

\begin{document}

\title[]{Modulated elliptic wave and asymptotic solitons in a shock problem to the modified Korteweg-de Vries equation}

\author{Vladimir Kotlyarov$^{\scriptsize 1\normalsize}$ and Alexander Minakov$^{\scriptsize 1,2,3\normalsize}$}
\address{$^{\scriptsize 1\normalsize}$ Mathematical Division, B. Verkin Institute for
Low Temperature Physics\\
47 Lenin Avenue, Kharkiv, 61103, Ukraine\\\vskip3mm
$^{\scriptsize 2\normalsize}$ Doppler Institute, Czech Technical University\\
Brehova 7, 11519 Prague, Czech Republic\\\vskip3mm $^{\scriptsize 3\normalsize}$Department of
Physics, Faculty of Nuclear Science and Physical
Engineering, Czech Technical University\\
Pohranicni 1288/1, Decin, Czech Republic\\\vskip3mm E-mail:
kotlyarov@ilt.kharkov.ua \\minakov.ilt@gmail.com
\vskip-1.5cm}

\begin{abstract}
We study the long-time asymptotic behavior of the Cauchy problem for the modified Korteweg - de Vries equation with an initial function of the step type. This function rapidly tends to zero as $x\to+\infty$ and to some positive constant $c$ as $x\to-\infty$.

In 1989 E. Khruslov and V. Kotlyarov have found \cite{KK} that for a large time the solution breaks up into a train of asymptotic solitons located in the domain
$4c^2t-C_N \ln t<x\leq4c^2t$ ($C_N$ is a constant). The number N of these solitons
grows unboundedly as $t\to\infty$. In 2010 V. Kotlyarov and A. Minakov have studied temporary asymptotics of the solution of the Cauchy problem on the whole line \cite{KM} and  have found that in the domain  $-6c^2 t<x<4c^2t$ this solution is described by a modulated elliptic wave.

We considere here the modulated elliptic wave in the domain $4c^2t-C_N \ln t<x<4c^2t$. Our main result shows that the modulated elliptic wave also breaks up into solitons, which are similar to the asymptotic solitons in \cite{KK}, but differ from them in phase. It means that the modulated elliptic wave does not represent the asymptotics of the solution in the domain $4c^2t-C_N \ln t<x<4c^2t$.  The  correct asymptotic behavior of the solution is given by the train of  asymptotic solitons given in \cite{KK}.

However, in the asymptotic regime as $t\to\infty$ in the region
$4c^2t-\dsfrac{N+1/4}{c}\ln t<x<4c^2t-\dsfrac{N-3/4}{c}\ln t$
we can watch precisely a pair of solitons with numbers $N$. One of them is the asymptotic soliton while the other soliton is generated from the elliptic wave. Their phases become closer to each other for a large $N$, i.e. these solitons are also close to each other. This result gives the answer on a very important question about matching of the asymptotic formulas in the mentioned region where the both formulas are well-defined. Thus we have here a new and earlier unknown  mechanism of matching of the asymptotics of the solution in the adjacent regions.
\end{abstract}



\section{Introduction}

The history of studying of the Cauchy problem for the modified Korteweg -- de Vries equation
\begin{equation}\label{mkdv}
q_t(x,t)+6q^2(x,t)q_x(x,t)+q_{xxx}(x,t)=0
\end{equation}
with an initial function of the step type
\begin{equation}\label{ic}
q(x,0)=q_0(x)\to\begin{cases}0\qquad {\rm as}\quad
x\to+\infty,\\c\qquad {\rm as }\quad x\to-\infty,\qquad
c>0,\end{cases}
\end{equation}
goes quite long. The first asymptotic results for a large time were obtained for the more famous Korteweg -- de Vries equation. Physicists have begun to understand a qualitative description of the solution beginning from the pioneer work of A.~Gurevich and L.~Pitaevsky \cite{GP} (1973). Later in this direction R.~Bikbaev,
V.~Novokshenov  and others had actively worked and obtained interesting results (sf. \cite{BikN1}-\cite{Bikb6}). In particular, the Cauchy problem (\ref{mkdv}) -- (\ref{ic}) with more general type of initial data was considered by  R.~Bikbaev \cite{Bikb4} in 1992. All these papers were done in the framework of the heuristic Whitham method. From the physical intuition it was understood that $x,t$-plane is divided into three domains. In the left and the right domains the solution tends to constants, one of them equals zero in our case.  In the middle domain it tends to a modulated elliptic wave.

There were not any rigorous mathematical papers on this theme with the exception for papers concerning to the  region of asymptotic solitons. It was proved that near the wave front there exists a domain of the strip type where the so-called asymptotic solitons arise. They are generated by a simple continuous spectrum of the corresponding Lax operator, as opposed to usual solitons generated by a discrete spectrum. For the KdV equation it was done by E.~Khruslov \cite{Kh1}, \cite{Kh2}
(1975, 1976), and for the MKdV it was done by E.~Khruslov and V.~Kotlyarov \cite{KK} (1989). The review of the results in this direction can be found in \cite{KK2}, \cite{KK3} and in the references therein. Thus, besides the three specified  regions,   there is an additional transition region near the leading edge where a train of  asymptotic solitons runs.

On the other hand, the method of the Riemann -- Hilbert problem and the corresponding steepest descent method \cite{DZ93} have been actively developed for more than 20 years. Recently these methods were successfully applied to studying of solutions of the  step type in other regions of $x,t$ half-plane, not only in the soliton domains (\cite{BK07} --\cite{BV},  \cite{EGKT}, \cite{KM} -- \cite{MK10}).

The Cauchy problem (\ref{mkdv})- (\ref{ic}) for the modified Korteweg -- de Vries equation was recently studied by V.~Kotlyarov and A.~Minakov \cite{KM} via the Riemann -- Hilbert approach. (More general initial data with nonzero backgrounds, defined by two different and nonzero constants as $x\to\pm\infty$, were studied  in \cite{M1}, \cite{M2}, \cite{KM2}.) In particular, in the domain $-6c^2 t<x<4c^2t $ the asymptotics of the  Cauchy problem solution is described by a modulated elliptic wave.

\begin{teor} {\rm\cite{KM}}\label{teorKM} In the region $-6c^2 t<x<4c^2t $
the solution of the IBV problem (\ref{mkdv}), (\ref{ic})  takes
form of a modulated elliptic wave
\[q(x,t)=q_{ell}(x,t)+\mathrm{o}\(1\),\quad t\rightarrow\infty,\] where
\begin{equation}{\label{qmod}}q_{ell}(x,t)=\sqrt{c^2-f^2(\xi)}
\dsfrac{\Theta\(\pi\i+\i t
B(f(\xi))+\i\Delta(f(\xi))|\tau(f(\xi))\)}{\Theta\(\i t
B(f(\xi))+\i\Delta(f(\xi))|\tau(f(\xi))\)}, \quad
\xi=\dsfrac{x}{12t}.
\end{equation}
Here $\Theta(z,\tau)$ is the Riemann theta function determined  by its $b$-period $\tau=\tau(d)$, the functions
 $B(d)$, $\tau(d)$, $\Delta(d)$ are explicitly defined via (\ref{Bg})-(\ref{Delta}), and the function $f(\xi)$  is defined implicitly through formulas  (\ref{dmu1}), (\ref{dmu2}).
\end{teor}

It is worth to mention that the point $\i d=\i f(\xi)$ ($0<d<c$)
is an analogue of a branch point of the Whitham zone which lies on the imaginary axis, $\tau$ is the $b$ -- period of the corresponding Riemann surface associated with the function $\sqrt{(k^2+c^2)(k^2+d^2)}$, the function $B(\xi)$ is the $b$ -- period of an Abelian integral of the second kind, and $\Delta(\xi)$  is an integral over the Riemann surface. All these ingredients are defined below in (\ref{dmu1}) -- (\ref{Delta}). It is easy to verify that the Riemann theta function is  ill-conditioned  if $d=f(\xi)$ tends to $c$ or $0$.  Therefore we cannot be sure that $ q_{ell}(x,t)$ gives the right asymptotics in a right semi-neighborhood of the trailing edge $x=-6c^2t$ and in a left semi-neighborhood of the leading edge $x=4c^2t$.

Nevertheless, the explicit formula, which determines $q_{ell}(x,t)$, is correctly defined in the domain $-6c^2t< x<4c^2t$. On the trailing edge one can easily find that $\lim\limits_{x\to-6 c^2t}q_{ell}(x,t)=c$. In the left region $x<-6c^2t$ the main term of the asymptotics of the solution of the IBV problem  is also equal to $c$ \cite{KM}. Hence, the main terms of the asymptotics  match at this point and there is no  need to introduce an additional asymptotic domain in the neighborhood of the trailing edge. Such a need may appear in the next order of an asymptotic expansion. Another situation is  in the neighborhood ($4c^2t-C_N \ln t<x<4c^2t$) of the leading edge where the train of asymptotic solitons runs \cite{KK}. Thus we must study the long-time asymptotic behavior of the modulated elliptic wave $q_{ell}(x,t)$ and to compare the result with the one earlier obtained for $q(x,t)$ in \cite{KK}. To do so we present here our main result.
\begin{teor}{\label{teormain}} Let $N\geq1$ be any integer number, and $\varepsilon>0$ is a fixed small number.
Then as $t\rightarrow\infty$\textbf{\\A)} uniformly for $t$ and
$x$ such that
$$4c^2t-\dsfrac{2N+\frac{1}{2}-\varepsilon}{2c}\ln t\leq
x\leq4c^2t-\dsfrac{\frac{1}{2}+\varepsilon}{2c}\ln t$$ the
following is true:
\begin{equation}{\label{qmodsumofasymptoticsolitons}}q_{ell}(x,t)=\sum_{n=1}^{N}
\dsfrac{2c}{\cosh\(2c(x-4c^2t)+\(2n-\frac{1}{2}\)\ln
t-\alpha_n(x,t)\)}+ \mathrm{O}\(\frac{1}{t}\);
\end{equation}
\textbf{B)} uniformly for $t$ and $x$ such that
$$4c^2t-\dsfrac{1-\varepsilon}{2c}\ln t\leq x<4c^2t$$
the following estimate holds:
\[q_{ell}(x,t)=
\mathrm{O}\(t^{-\frac{1}{2}}\).
\]
Here $\alpha_n(x,t)$ is a well-defined function, which has the
following asymptotic behavior as $t\to\infty:$
\begin{equation}{\label{alpha_nxt}}\alpha_n(x,t)=
-\(2n-\frac{1}{2}\)\ln\dsfrac{\ln\frac{1}{v}}{vt}-\dsfrac{8c^3tv}{\ln
\frac{1}{v}}-\(6n-\dsfrac{7}{2}\)\ln2-2\ln|h^*|+\mathrm{O}\(\dsfrac{\ln^2\ln\frac{1}{v}}{\ln
v}\),\end{equation} where
$v=1-\dsfrac{x}{4c^2t}$, and $h^*$ is a constant determined by  the initial data and defined by the associated transmission coefficient by formula (\ref{asingularity}).
For any $\gamma, \delta$ such that $0<\delta<\gamma$, this
function $\alpha_n(x,t)$ is bounded in the interval
$4c^2t-\gamma\ln t<x<4c^2t-\delta\ln t$.
\end{teor}

\begin{com} In the intersection domain
$4c^2t-\dsfrac{1-\varepsilon}{2c}\ln t\leq
x\leq4c^2t-\dsfrac{\frac{1}{2}+\varepsilon}{2c}\ln t$
asymptotics \textbf{A)} and \textbf{B)}  match.
\end{com}
\begin{com}
On any curve $x=4c^2t -\gamma \ln t$ we can express
$\alpha_n(x,t)$ as a sum of some constant and a  remainder  which tends
to zero as $t\rightarrow\infty.$ But this constant depends on
$\gamma$, and that is why we can not express $\alpha_n(x,t)$ as a
sum of some constant and a vanishing term in any domain $4c^2t
-\gamma \ln t<x<4c^2t -\delta \ln t.$
\end{com}

Further we compare $q_{ell}(x,t)$
(\ref{qmodsumofasymptoticsolitons}) with the asymptotic representation (see below (\ref{qsumofasymptoticsolitons}) - (\ref{tildealpha}))  for the solution $q(x,t)$ of the IBV problem (\ref{mkdv})-(\ref{ic}) that follows from \cite{KK}:

\begin{teor} \label{teorKK}
Let $N\geq1$ be an integer number.
Then for
$$x>4c^2t-\dsfrac{(2N+1)\ln t}{2c}$$
the solution of the IBV problem (\ref{mkdv})-(\ref{ic}) is presented as follows:
\begin{eqnarray}{\label{qsumofasymptoticsolitons}}
q(x,t)&=q_{as}(x,t)+\mathrm{O}\left(t^{-\frac{1}{2}+\varepsilon}\right),\\
q_{as}(x,t)&=\sum\limits_{n=1}^N\dsfrac{2c}{\cosh\left(2c(x-4c^2t)+\left(2n-\dsfrac{1}{2}\right)
\ln t-\tilde\alpha_n\right)} \label{qass},
\end{eqnarray}
\begin{equation} {\label{tildealpha}}
\tilde\alpha_n=\ln\left[\dsfrac{|h_0|}{4^{2n-1}(2c)^{6n-3}}
\Gamma(n)\Gamma\(n+\dsfrac{1}{2}\)\right], \end{equation}
where $\Gamma(n)$, $\Gamma(n+1/2)$ are the Gamma-functions, and  $h_0$ is a constant determined by  the initial data and defined by the associated transmission coefficient by  formulas (\ref{h0}), (\ref{|h0|}).
\end{teor}

In order to compare the solitons from (\ref{qmodsumofasymptoticsolitons}) with asymptotic solitons from (\ref{qass}), we compare their phases (\ref{alpha_nxt}), (\ref{tildealpha}) on the curve $x_n(t)$ (\ref{x_n(t) definition implicit}) of the peaks of the modulated elliptic wave solitons:
$$
\beta_n(t):=\alpha_n\left(x_n(t), t\right)-\tilde\alpha_n,
$$
which characterizes a deviation of the solitons from each other.
We prove that in the  asymptotic regime, when  $t\to\infty$ and $n\rightarrow\infty$,  these deviation becomes smaller and smaller.
\begin{teor} \label{betan}
The following estimate of the asymptotic smallness takes place:
$$
\beta_n(t)=\mathrm{O}\left(\frac{1}{2n}\right) +\mathrm{O}\left(\frac{\ln^2\ln t}{\ln t}\right) \qquad \mathrm{as}\quad  n,t\to\infty.
$$
\end{teor}
This result gives the answer on an old question that  appeared on the Marchenko's seminar in the late 80-s for the  Korteweg -- de Vries equation about matching of  asymptotic formulas in a neighborhood of the leading wave front. This question  arose  out of works  \cite{BikN1}, \cite{Kh2} and remained open for any nonlinear model. For the modified Korteweg de Vries equation we state here that  the  modulated elliptic wave is unsuitable  in the region $4c^2t-\dsfrac{N+1/4}{c}<x<4c^2t$ with any natural $N$, while the  asymptotic solitons  (\ref{qass})  give the right asymptotic behavior of the solution in this region.  This statement is the second main result of the present paper.
In our mind  these results are also very important for the steepest descent method in the theory of the matrix Riemann--Hilbert problem. They bring a new and earlier unknown  mechanism of matching of asymptotics of the solution in adjacent regions.

The paper is organized as follows.
In section 2 we introduce some notations needed in the sequel.
In section 3 we give the proof of Theorem \ref{teormain}. The proof is based on some quite technical lemmas, which are posted  in Appendices 1 and 2 (section 6 and 7).  In section \ref{Comparing} we deduce Theorem \ref{teorKK} and prove that the deviation  $\beta_n(t)$ is not asymptotically small for any finite $n$. In section \ref{Matching} we prove the estimate of smallness for deviation  $\beta_n(t)$ for large $n$ and $t$ (Theorem \ref{betan}).

\section{Preliminaries}

Let us consider the initial-value prob\-lem (\ref{mkdv}),
(\ref{ic}) with an initial func\-tion which rapidly tends to its
limits, namely
\begin{equation}\label{exp_decreasing}\int\limits_{-\infty}^0|q_0(x)-c|\e^{2|x|l_0}\d
x+\int\limits_{0}^{+\infty}|q_0(x)|\e^{2xl_0}\d
x<\infty.\end{equation} Here $l_0$ satisfies the inequality $l_0>c>0.$
We suppose that there exists the solution of this initial
value-problem and that this solution converges to its limits with
the first moment: for any $t$
\begin{equation}\label{first_moment}\int\limits_{-\infty}^0|x||q(x,t)-c|\d
x+\int\limits_{0}^{+\infty}x|q(x,t)|\d x<\infty.
\end{equation}
 We are interested in the long-time asymptotic behavior of the solution of the Cauchy problem.

Let $a^{-1}(k)$ be the standard transmission coefficient
(\cite{KM}). It has the following properties:

\begin{enumerate}
\item $a(k)$ is analytic in the strip $\left\{k:\quad  0<\Im k<l_0\right\}$
with the cut $[0,\i c]$, can be extended continuously up to the boundary
with the exception of the point $\i c,$ where $a(k)$ may have at
most a root singularity of the order $(k-\i c)^{-1/4};$

\item $a(k)$ satisfies the symmetry condition
\begin{equation}\label{asymmetry}\ol{a(-\ol{k})}=a(k);\end{equation}

\item $\forall k\in(\i c,0] \ \ \ a_{\pm}(k)\neq0,$ \\where
$a_{+}(k)$ and $a_-(k)$ are boundary values of the function $a(k)$
from the right and from the left of the interval $(\i c,0)$,
respectively.

\end{enumerate}

\noindent We suppose the absence of usual solitons generated by the discrete spectrum:
$$\forall
k\in\mathbb{C}_+\backslash[\i c,0] \ \ \ a(k)\neq0.$$
We suppose also that $a(k)$ has a root singularity at the point $\i c$ of the  exactly order
$(k-\i c)^{-1/4}$:
\begin{equation}\label{asingularity}a(k)=\dsfrac{h^*}{2}\sqrt[4]{\dsfrac{2\i c}{k-\i
c}}\(1+\mathrm{O}\(\sqrt{\dsfrac{k-\i c}{\i}}\)\),\ k\to\i
c,\qquad \qquad h^*\in\mathbb{R}\backslash
\left\{0\right\}.\end{equation}

\noindent The fact that $h^*\in\mathbb{R}$ follows from formula
$(\ref{asymmetry}).$

\begin{com} Particularly, $a(k)$ has the such type of
singularity with $h^*=1$ in the case of the "pure" step initial function
$q_0(x)=\widetilde{q}_0(x)$, where $\widetilde{q}_0(x)$ is a
Heaviside -- type function
\begin{equation}\nonumber \widetilde{q}_0(x)=\begin{cases}0,\quad
x> 0,\\c,\quad x<0.\end{cases}
\end{equation}
In this case the inverse of the transmission coefficient
$a(k)=\widetilde{a}(k)$ can be computed explicitly and has the
following form:
$$\widetilde{a}(k)=\dsfrac{1}{2}\(\sqrt[4]{\dsfrac{k-\i c}{k+\i
c}}+\sqrt[4]{\dsfrac{k+\i c}{k-\i c}}\),$$ where the cut is taken
across the segment $[\i c,-\i c]$ and a branch of the root is
taken such that it tends to $1$ as $k\rightarrow\infty.$

Due to the analyticity of $\widetilde{a}(k)$, the asymptotic analysis
in this case can be done in the same manner as in the case of smooth
and fast decreasing  initial functions.
\end{com}

Now we introduce the notations needed in the sequel.

$\bullet$ \textbf{$\Theta$-function $\Theta(z|\tau)$}.\\ We use
$\Theta$-function in the next form:
\\let $\tau<0$ be a fixed negative number. Then for any $z\in\mathbb{C}$
\begin{equation}{\label{Theta}}\Theta(z|\tau)=\sum\limits_{m=-\infty}^{\infty}\exp\left\{\dsfrac{1}{2}\tau m^2
+zm\right\}.\end{equation}
$\Theta$-function has the following property:\\
for any integer $n, l$
\begin{equation}{\label{Thetashift}}
\Theta(z+2\pi\i n+\tau
l|\tau)=\Theta(z|\tau)\exp\left\{-\dsfrac{1}{2} \tau
l^2-zl\right\}.\end{equation}

\begin{com} In the sequel the parameter   $\tau$  is  a function on $\xi$, which is defined via formulas (\ref{tau}), (\ref{dmu1}), (\ref{dmu2}), wherein  $\tau(\xi)$   tends to $0$ as $\xi\rightarrow c^2/3$.\end{com}

$\bullet$ \textbf{Real parameters $d$ and $\mu $}.\\
Function
$d=f(\xi):\left[-\frac{c^2}{2},\frac{c^2}{3}\right]\rightarrow[0,c]$
is defined from the following system:
\begin{equation}{\label{dmu1}}
\int\limits_0^1(\mu^2-\lambda^2d^2)\sqrt{\dsfrac{1-\lambda^2}{c^2-\lambda^2d^2}}\d\lambda=0,
\end{equation}
\begin{equation}{\label{dmu2}}
\dsfrac{c^2}{2}+\xi=\mu^2+\dsfrac{d^2}{2}.
\end{equation}
This system was introduced in \cite{KM}.

\begin{com} Usually in asymptotic problems $\xi=\dsfrac{x}{12t}$ is referred to as the  "slow" variable. That means that the
asymptotic analysis of the solution of the initial value problem is studied for those $\xi$ which lie in any fixed bounded  interval while $t\rightarrow\infty$.
\end{com}

\begin{lem}{\label{lemdmu}}
For $\xi\in\left[\frac{-c^2}{2},\frac{c^2}{3}\right]$  there
exists  one to one, increasing and continuous map
$f:\left[\frac{-c^2}{2},\frac{c^2}{3}\right]\rightarrow[0,c],$
which defines the  unique solution of the system of equations
(\ref{dmu1}), (\ref{dmu2}) by the formulas $d=f(\xi)$ and
$\mu=\mu(f(\xi)) $. The map
$\mu:[0,c]\rightarrow\left[0,\frac{c}{\sqrt{3}}\right]$ is also
continuous and increasing, and the following boundary conditions
are satisfied: $f\(\frac{c^2}{3}\)=c$,
$f\(\frac{-c^2}{2}\)=0$, $\mu(0)=0$,
$\mu(c)=\frac{c}{\sqrt{3}}.$
\end{lem}
The proof of lemma \ref{lemdmu} is given in Appendix 2 (subsection 7.2).
\\\\ $\bullet$ \textbf{Functions $B(d), \tau(d),
\Delta(d)$}.
\\For $d\in(0,c)$ let $\w(k,d)=\sqrt{(k^2+c^2)(k^2+d^2)}$ be the analytic in $k$ function in the domain $k\in\mathbb{C}\backslash\([\i c,\i d]\cup[-\i
c,-\i d]\)$, where the square root is fixed by the condition
$\w(0,d)=cd$. Then we define $B(d), \tau(d), \Delta(d)$ as
follows:

\begin{equation} {\label{Bg}}B(d)=24\ds\int\limits_{\i d}^{\i
c}\dsfrac{(k^2+\mu^2(d))(k^2+d^2)\d k}{\w_+(k,d)},
\end{equation}

\begin{equation}{\label{tau}}\tau(d)=-\pi\i\ds\int\limits_{\i d}^{\i c}\dsfrac{\d k}{\w_+(k,d)}
\(\ds\int\limits_0^{\i d}\dsfrac{\d k}{\w(k,d)}\)^{-1},
\end{equation}

\begin{equation}{\label{Delta}}\Delta(d)=
 \int\limits_{\i d}^{\i c}\dsfrac{\ln\left(a_+(k)a_-(k)\right)\d k}{\w_+(k,d)}
\left(
 \i\int\limits_{0}^{\i d}\dsfrac{\d k}{\w(k,d)}\right)^{-1} .
\end{equation}

Here by $a_+(k)$, $\w_+(k)$ and $a_-(k)$, $\w_-(k)$ we denote the limit
values of $a(k), \w(k)$ taken from the right and from the left side of
the interval $[\i c,-\i c]$, respectively.

\begin{com} It can be verified that $\tau(d)<0.$ That is why
we can take $\tau(d)$ or $\dsfrac{4\pi^2}{\tau(d)}$ as the $\tau$
- parameter of $\Theta$ - function (\ref{Theta}). Defined by
(\ref{tau}) function $\tau(d)$ is indeed the $\tau$ - parameter of
$\Theta$ - function in formula (\ref{qmod}) from Theorem
\ref{teorKM}.
\end{com}

\begin{com} Let us explain some inconvenience in our notations
and the reasons why we can not make them more convenient.

First of all, we have a parameter $\xi=\dsfrac{x}{12t},$ the
so-called "slow" variable.

Secondly, we have a Riemann surface, which is associated with function $\w(k)=\sqrt{(k^2+c^2)(k^2+d^2)},$ $0<d<c.$

In the third place, we have the quantity $\tau$, which is defined by formula (\ref{tau}) and which is actually the $b$ -- period of this Riemann surface.

We will also use the quantity $\tau^*=\dsfrac{4\pi^2}{\tau}$, which is the $b$ - period of the same Riemann surface but with another homologic basis.

Finally, for our asymptotic analysis it is convenient to introduce the new "small" variable $\eta=1-\dsfrac{d}{c}.$

All  these quantities are in the one-to-one  correspondence.
But sometimes in the sequel  it is  convenient to consider some quantities as functions in $\xi,$ and sometimes as functions in $d,$ or $\tau,$ or $\tau^*$, or $\eta.$

That is why we need different notations for the same  quantity when we consider it as a function of different variables.

So, for $d$ as a function in $\eta$, we use  the  notation $d=d(\eta)$.
For $d$ as a function in $\xi$, we use the notation $d=f(\xi).$
Finally, for $d$ as a function in $\tau^*$, we use the  notation
$d=h(\tau^*).$
\end{com}

\section{Proof of Theorem \ref{teormain}}

Series which represent the theta-functions in (\ref{qmod}) are
slowly-convergent as $\xi$ tends to $\frac{c^2}{3},$
since in this case $\tau(f(\xi))$ tends to 0 (see lemma \ref{h}
in appendix 2). To overcome these difficulties, i.e. to make the theta
series well converged, we can use Poisson summation formula and
rewrite (\ref{qmod}) in terms of a rapidly convergent theta series.
Indeed, we have
\begin{equation}{\label{PoissonSum}}
\Theta(z\vert\tau)=\Theta\left(\displaystyle\frac{2\pi i
z}{\tau}\Big\vert\displaystyle\frac{4
\pi^2}{\tau}\right)\sqrt{\displaystyle\frac{2\pi}{-\tau}}
\left(\exp{\displaystyle\frac{-z^2}{2\tau}}\right).
\end{equation}
Then the model approximation (\ref{qmod}) of the solution takes
the form:
\begin{eqnarray}
 q_{ell}(12t\xi,t)
&=\sqrt{c^2-f^2(\xi)}\exp{\left(\displaystyle\frac{-\tau^*(f(\xi))}{8}+
\displaystyle\frac{\tau^*(f(\xi))}{4}\left(z(t,\xi)+1\right)
\right)}\\&\ \times\displaystyle\frac{
\Theta\left(\displaystyle\frac{\tau^*(f(\xi))}{2}\left(
z(t,\xi)+1\right)\Big\vert\tau^*(f(\xi))\right)}{\Theta\left(\displaystyle
\frac{\tau^*(f(\xi))}{2}z(t,\xi)
\Big\vert\tau^*(f(\xi))\right)}\label{qmod*},
\end{eqnarray}
where \begin{equation}\label{z(t,xi)}z(t,\xi)=\dsfrac{1}{\pi}\(tB(f(\xi))+\Delta(f(\xi))\)\end{equation} and
$\tau^*(d)=\dsfrac{4\pi^2}{\tau(d)}.$ We emphasize that now
$\tau^*(f(\xi))\to -\infty$ as $\xi\to\dsfrac{c^2}{3}$ and hence the theta series becomes well-converged.

Let us introduce the function F(.,.). For any $\tau*<0$ and
$z\in\mathbb{C}$ we set
\begin{equation}{\label{F}}
F(\tau^*,z)=\sqrt{c^2-h^2(\tau^*)}\exp\(\dsfrac{-\tau^*}{8}+\dsfrac{\tau^*}{4}(z+1)\)
\dsfrac{\Theta\(\dsfrac{\tau^*}{2}(z+1)\vert\tau^*\)}
{\Theta\(\dsfrac{\tau^*}{2}z\vert\tau^*\)}.
\end{equation}
Here function $h(.):(-\infty,0)\rightarrow(0,c)$ is the inverse function of the  function $\tau^*$: $\quad h(\tau^*(d))=d$  (see lemma \ref{h} in Appendix 2). Thus we can rewrite formula (\ref{qmod*}) by using the function $F(.,.):$
\begin{equation}\label{qelinF}q_{ell}(x,t)=F(\tau^*(f(\xi)),z(t,\xi)).\end{equation}

\noindent The following lemma describes some properties of
$F(.,.):$

\begin{lem}{\label{lemmaF}}
For any integer $N\geq1$ the following formulas
\begin{equation}\label{F(tau*,z)}F(\tau^*,z)=\sum_{n=1}^{N}\dsfrac{2c}{\cosh\frac{\tau^*(-1-z+2n)}{4}}+\mathrm{O}\(\e^{\frac{\tau^*}{4}}\),\quad
0\leq z\leq2N,\end{equation}
\[F(\tau^*,z)=\mathrm{O}\(\e^{\frac{\tau^*}{8}}\),\quad
\dsfrac{-1}{2}\leq z\leq\dsfrac{1}{2}\]
hold uniformly in $z$ as $\tau^*\rightarrow-\infty$.
\end{lem}
\noindent Proof of this lemma is given in Appendix 1.

Now we have to express the intervals $z\in[0,2N]$ and
$z\in[-\frac{1}{2},\frac{1}{2}]$ in terms of variables $x,t$. To
do this let us define  the   small parameters $\eta$ and $v$:
\[
\eta=\eta(d)=1-\ds\frac{d}{c}, \qquad \textrm{and} \qquad
v=v(\xi)=1-\ds\frac{3\xi}{c^2},
\]
which tend to zero as $\xi\to c^2/3$. The inverse of these
functions are
\[d(\eta)=c(1-\eta), \qquad \textrm{and} \qquad
\xi(v)=\dsfrac{c^2(1-v)}{3}.\] In Appendix 2 we found the
connection between $v$ and $\eta$ (see formulas (\ref{vineta}),
(\ref{etainv})),

\begin{equation}\nonumber\dsfrac{v}{8\e}=\dsfrac{\eta}{8\e}\ln\dsfrac{8\e}{\eta}
+\ord\eta^2\ln^2\eta\right),\quad \eta\rightarrow+0.\end{equation}

\begin{eqnarray}\nonumber
&\eta=\dsfrac{v}{\ln\dsfrac{1}{v}}\left(1-\frac {\ln
8\e+\ln\ln\dsfrac{1}{v}} {\ln\dsfrac{1}{v}}+
\ord\dsfrac{\ln^2\ln\dsfrac{1}{v}}{\ln^2\dsfrac{1}{v}}\right)\right)
,\quad v\rightarrow+0.
\end{eqnarray}
We also found the   asymptotic expansions of $B(d(\eta))$,
$\Delta(d(\eta))$, $\tau^*(d(\eta))$ in $\eta$
(see formulas
(\ref{Bgas}), (\ref{Deltaineta}), (\ref{tau^*})).
\begin{equation}B(d(\eta))=8\pi
c^3\eta(1+\ord\eta\ln\eta\right)),\quad\eta\rightarrow0.\nonumber
\end{equation}
\begin{eqnarray}\nonumber
&\frac{1}{\pi}\Delta(d(\eta))=\dsfrac{-1}{2}\left(1-\dsfrac{4\ln\dsfrac{2}{|h^*|}}{\ln\dsfrac{1}{\eta}}+
\ord\dsfrac{1}{\ln^2\eta}\right)\right).
\end{eqnarray}
\begin{eqnarray*}
\tau^*=\tau^*(d(\eta))=-4\ln\frac{8}{\eta}+\ord\eta\right).
\end{eqnarray*}

Thus we
come to the following lemma, the detailed proof is given in Appendix 1:\\
\begin{lem}{\label{lemmaz}}
\begin{equation}{\label{z}}z(t,\xi(v))=\dsfrac{8c^3tv}{\ln\dsfrac{1}{v}}-\dsfrac{1}{2}+
\mathrm{O}\(\dsfrac{1}{\ln\frac{1}{v}}+\dsfrac{tv\ln\ln\frac{1}{v}}{\ln^2v}\),\qquad {\textrm as}\quad v\to+0.
\end{equation}
\end{lem}
\bigskip

\textbf{The correspondence between intervals in $z$ and intervals in
$x,t$.}
\begin{lem}{\label{lemmaintervals}}\textbf{\\A. }
For any integer $N\geq1$ and any sufficiently small positive
$\varepsilon$ there exists $T=T(\varepsilon, N)$ such that for any
$x,t$ which satisfy the condition $t>T$ and  the condition
\begin{equation}\nonumber{\label{xtinterval}}4c^2t-\dsfrac{2N+\frac{1}{2}-\varepsilon}{2c}
\ln t<x<4c^2 t-\dsfrac{\varepsilon+\frac{1}{2}}{2c}\ln t
\end{equation}
the following formula is true:
\[0\leq z(t,\xi)\leq2N.\]
\textbf{\\B. } For any integer $N\geq1$ and any sufficiently small
positive $\varepsilon$ there exists $T=T(\varepsilon, N)$ such
that for any $x,t$ which satisfy the condition $t>T$ and condition
\begin{equation}\nonumber{\label{xtinterval}}4c^2t-\dsfrac{1-\varepsilon}{2c}
\ln t<x<4c^2t
\end{equation}
the following formula is true:
\[-\dsfrac{1}{2}\leq z(t,\xi)\leq\dsfrac{1}{2}.\]
\end{lem}

\noindent Proof of this lemma is given in Appendix 1.\\

\textbf{Evaluating of cosh argument in formula (\ref{F(tau*,z)}) for $q_{ell}$}
\\Define $\alpha_n(x,t)$ by the relation
\begin{equation}\label{argument in cosh}\dsfrac{\tau^*(z-2n+1)}{4}\equiv2c(x-4c^2t)+\(2n-\frac{1}{2}\)\ln t-\alpha_n(x,t).\end{equation}
Here $\tau^*\equiv\tau^*(f(\xi))$, $z\equiv z(t,\xi)$ (\ref{z(t,xi)}), $x\equiv
12t\xi.$
 Then we have the following lemma
\begin{lem}{\label{lemmaarg}}
$\alpha_n(x,t)=$\\
$$-\(2n-\frac{1}{2}\)\ln\dsfrac{\ln\frac{1}{v}}{vt}-\dsfrac{8c^3tv}{\ln
\frac{1}{v}}-\(6n-\dsfrac{7}{2}\)\ln2-2\ln|h^*|+\mathrm{O}\(\dsfrac{\ln^2\ln\frac{1}{v}}{\ln
v}\).$$
The function $\alpha_n(x,t)$ is bounded for $4c^2t-\gamma\ln
t<x<4c^2t-\delta\ln t$, where $0<\delta<\gamma$ are arbitrary
numbers.
\end{lem}
\noindent Proof of this lemma is given in Appendix 1.\\

\noindent By summarizing lemmas \ref{lemmaF}, \ref{lemmaintervals}, \ref{lemmaarg},  and formulas (\ref{qelinF}), (\ref{argument in cosh}),  we obtain the statements of the Theorem
\ref{teormain}. Q. E. D.

\section{Comparing of $q_{ell}(x,t)$ and $q(x,t)$ in the domain $4c^2t-\gamma\ln t<x\leq 4c^2t.$}\label{Comparing}
In this section we compare the asymptotics for the solution $q(x,t)$ of the ibv (\ref{mkdv})-(\ref{ic})  with $q_{ell}(x,t)$
(\ref{qmodsumofasymptoticsolitons}).

\textbf{Theorem} \cite{KK} \textit{Let $N\geq1$ and
$M$ be such integer numbers that
$N=\left[\dsfrac{M+1}{2}\right]$. Then for
$$x>4c^2t-\dsfrac{(M+1)\ln t}{2c}$$
the solution of the IBV problem (\ref{mkdv})-(\ref{ic}) is presented as follows
\begin{eqnarray}\nonumber
&q(x,t)=
\sum\limits_{n=1}^N\dsfrac{2c}{\cosh\left(2c(x-4c^2t)+\left(2n-\dsfrac{1}{2}\right)
\ln
t-\tilde\alpha_n\right)}+\mathrm{O}\left(t^{-\frac{1}{2}+\varepsilon}\right),\\\nonumber
&\mathrm{where\ }
\tilde\alpha_n=\ln\left[\dsfrac{|h_0|}{(n-1)!^24^{2n-1}(2c)^{6n-3}}\dsfrac{\Gamma^{(n)}_1\Gamma^{(n)}_{3\slash2}}
{\Gamma^{(n-1)}_1\Gamma^{(n-1)}_{3\slash2}}\right],\\\nonumber
&\Gamma^{(k)}_{b}=\det\left[\left[\Gamma(i+j+b)\right]\right]_{i,j=\overline{0,k-1}}.\nonumber
\end{eqnarray}
Here $h_0$ is some constant, which is determined by the initial function $q_0(x)$, and $\Gamma(.)$ is the Gamma-function.
}

First of all we reformulate the theorem by simplifying expression for
$\tilde\alpha_n$. The determinants $\Gamma^{(k)}_{b}$ are calculated explicitly and one can obtain that
$$
\dsfrac{\Gamma^{(n)}_1\Gamma^{(n)}_{3\slash2}}
{\Gamma^{(n-1)}_1\Gamma^{(n-1)}_{3\slash2}}=(n-1)!^2\Gamma(n)\Gamma(n+1/2).
$$
Now without loss of generality we can choose $M=2N$. Then we easily obtain the statement of Theorem \ref{teorKK}.

Let us consider  the  curve
$x=4c^2t-\dsfrac{2m-\frac{1}{2}}{2c}\ln
t+\dsfrac{\tilde\alpha_m}{2c},$ where $m\in\mathbb{N}.$ All the
summands in formula (\ref{qsumofasymptoticsolitons}) are exponentially
small except for  the   one with $n=m$. Then  for $x,t$ which lie on
this curve we have
\[q_{as}(x,t)=2c.\]

Now let us consider $q_{ell}(x,t)$. All the summands in formula
(\ref{qmodsumofasymptoticsolitons}) are exponentially small  except for  the
one with $n=m$. According to formula (\ref{alpha_nxt}), for
$x=4c^2t-\dsfrac{2m-\frac{1}{2}}{2c}\ln
t+\dsfrac{\tilde\alpha_m}{2c}$ and $t\to\infty$ we have the
following:

$$\dsfrac{vt}{\ln \frac{1}{v}}=\dsfrac{2m-\frac{1}{2}}{8c^3}\(1+\mathrm{O}\(\dsfrac{\ln\ln t}{\ln t}\)\),$$
$$\ln\dsfrac{vt}{\ln \frac{1}{v}}=\ln\dsfrac{2m-\frac{1}{2}}{8c^3}+\mathrm{O}\(\dsfrac{\ln\ln t}{\ln t}\),$$
and
\begin{equation}\label{alpha_alpha_tilde} \alpha_m(x,t)=\alpha_m+\mathrm{O}\(\dsfrac{\ln^2\ln t}
{\ln t}\),\end{equation} where $\alpha_m$ is defined by the formula
\begin{equation}{\label{alpha}}\alpha_m:=\(2m-\dsfrac{1}{2}\)\ln\dsfrac{2m-\frac{1}{2}}{8c^3\e}-\(6m-\dsfrac{7}{2}\)\ln
2-2\ln|h^*|.\end{equation} Then
$$q_{ell}(x,t)=\dsfrac{2c}{\cosh\(\tilde\alpha_m-\alpha_m(x,t)\)}=
\dsfrac{2c}{\cosh\(\tilde\alpha_m-\alpha_m\)}+\mathrm{O}\(\dsfrac{\ln^2\ln
t}{\ln t}\).$$ Since $\alpha_m$ (\ref{alpha}) and $\tilde\alpha_m$
(\ref{tildealpha}) are not equal, we conclude that $q_{as}(x,t)-q_{ell}(x,t)$ does
not tend to 0 as $t\to\infty$.

Thus the  modulated elliptic wave is unsuitable  in the region $4c^2t-\dsfrac{N+1/4}{c}<x<4c^2t$ with any natural $N$. Instead, the  asymptotic solitons (\ref{qass}), (\ref{tildealpha}) give the right asymptotic behavior of the solution in this region.

\section{Matching}\label{Matching}

In this section we prove Theorem \ref{betan}.
\begin{proof}
The   asymptotic solitons from \cite{KK} take the form
$$
s_{as}(x,t)=\frac{2c}{\cosh[2c(x-4c^2t+\frac{1}{2c}\ln t^{2n-1/2}-\frac{\tilde\alpha_n}{2c})]} ,
$$
where (see Theorem \ref{teorKK})
$$
\tilde\alpha_n=\ln\left[\dsfrac{|h_0|}{4^{2n-1}(2c)^{6n-3}}
\Gamma(n)\Gamma\(n+\dsfrac{1}{2}\)\right].
$$

Now let us consider the phases in details.
Indeed, what is $h_0$?
From \cite{KK},   for   a  general step-like function we have
\begin{equation}\label{h0}
h(k(\rho))=\rho^{1/2}(h_0+h_1\rho^{1/2}+\ldots),
\end{equation}
where $\rho=2\nu(C +12\mu^2-4\nu^2)$, $k=\mu+\i\nu$, $C =4 {\rm Im}^2 E-12 {\rm Re}^2 E$.  In our case $E=\i c, \mu=0.$

Hence, we have
$$
h(k)=\dsfrac{1}{2\pi}f_+(k)=\dsfrac{\i}{2\pi a_-(k)a_+(k)}.
$$
Therefore, taking into account (\ref{asingularity}),
$$h(k)=\dsfrac{1}{2\pi {h^*}^{2}}\dsfrac{2\i}{c}\sqrt{k^2+c^2}\left[1+\mathrm{O}
(\sqrt{k-\i c})\right]=\dsfrac{\i}{\pi {h^*}^{2}c}\sqrt{k^2+c^2}\left[1+\mathrm{O}
(\sqrt{k-\i c})\right].
$$
  On the other hand, since $k=\i\nu$, $E=\i c$, we get $\rho=8\nu(c^2-\nu^2)$ and
$$
\sqrt{k^2+c^2}=\sqrt{\dsfrac{\rho}{8c}}(1+\dsfrac{\rho}{32c^3}+\ldots).
$$
Hence
$$
h(k(\rho))=\dsfrac{\i}{\pi {h^*}^{2}(2c)^{3/2}}\sqrt{\rho}[1+\mathrm{O}(\sqrt{\rho})],
$$
Thus
\begin{equation}\label{|h0|}
|h_0|=\dsfrac{1}{{h^*}^2\pi(2c)^{3/2}}.
\end{equation}
Therefore
$$
\tilde\alpha_n=\ln\left[\dsfrac{|h_0|}{4^{2n-1}(2c)^{6n-3}}
\Gamma(n)\Gamma\(n+\dsfrac{1}{2}\)\right]=\ln\left[\dsfrac{\Gamma(n)\Gamma\(n+\dsfrac{1}{2}\)}
{(h^*)^2\pi 2^{4n-2}(2c)^{6n-3/2}}\right].
$$
Since  $\Gamma(n)=(n-1)!$ and $\Gamma(n+1/2)=\dsfrac{\sqrt{\pi}(2n)!}{2^{2n} n!}$
(sf.\cite{Lavrentiev Sabat}), we have
\begin{equation}\label{alpha_tilde_final}
\tilde\alpha_n=\ln\left[\dsfrac{{(h^*)}^{-2}(2n)!}{2^{12n-7/2}(c)^{6n-3/2} \sqrt{\pi}\, n}\right]
\end{equation}

\noindent
The elliptic wave (\ref{qmod}) is the sum of the elliptic asymptotic solitons (\ref{qmodsumofasymptoticsolitons})
$$
s_{ell}(x,t)=\frac{2c}{\cosh[2c(x-4c^2t+\frac{1}{2c}\ln t^{2n-1/2}-\frac{\alpha_n(x,t)}{2c})]},
$$

Let $x=x_n(t)$ be such a curve that $s_{ell}(x_n(t),t)\equiv2c$. Then $x_n(t)$ is the solution of the equation
\begin{equation}\label{x_n(t) definition implicit}
x_n(t)=4c^2t-\frac{1}{2c}\ln t^{2n-1/2}+\frac{\alpha_n(x_n(t),t)}{2c}.
\end{equation}
We look for the solution in the form
$$
x_n(t)=4c^2t-\frac{1}{2c}\ln t^{2n-1/2}+\frac{z_n(t)}{2c}.
$$
Hence
$$
z_n(t)= \alpha_n(x_n(t),t).
$$
It is shown in section \ref{Comparing}, formulas (\ref{alpha_alpha_tilde}), (\ref{alpha}), that for any bounded in $t$ function $z_n(t)$
$$
\alpha_n(x_n(t),t)=\alpha_n+\mathrm{O}\left(\frac{\ln^2\ln t}{\ln t}\right),
$$
 where  $$
\alpha_n=(2n-1/2)\ln\left(\dsfrac{2n-1/2}{8c^3\e}\right)-(6n-7/2)\ln 2-2\ln|h^*|,
$$
as $t\to\infty$. Therefore $z_n(t)=\alpha_n+\mathrm{O}\left(\frac{\ln^2\ln t}{\ln t}\right).$
\\
On the other hand,
\begin{eqnarray*}
\alpha_n&=(2n-1/2)\ln\left(\dsfrac{2n-1/2}{8c^3\e}\right)-(6n-7/2)\ln 2-2\ln|h^*|
=\\&=\ln\left[\left(\dsfrac{2n-1/2}{\e}\right)^{2n-1/2}\dsfrac{(h^*)^{-2}}{2^{12n-5}(c)^{6n-3/2} }\right],
\end{eqnarray*}
and hence  (\ref{alpha_tilde_final})
$$
\alpha_n-\tilde\alpha_n=\ln\left[\left(\dsfrac{2n-1/2}{\e}\right)^{2n-1/2}
\dsfrac{n\sqrt{\pi}\, 2^{3/2}}{(2n)! }\right].
$$
For a finite $n$ the phase difference is not equal to zero and    the  elliptic asymptotic solitons do not coincide with the  asymptotic solitons.
But for a large $n$, we have
$$
(2n-1/2)^{2n-1/2}=(2n)^{2n-1/2}(1-1/4n)^{4n/2-1/2}=
(2n)^{2n-1/2}\e^{-1/2}[1+\mathrm{O}(1/4n)],
$$
$$
(2n-1/2)^{2n-1/2}= \dsfrac{(2n)^{2n}}{\sqrt{2n\e}} [1+\mathrm{O}(1/4n)],
$$
i.e.
$$
\left(\dsfrac{2n-1/2}{\e}\right)^{2n-1/2}= \left(\dsfrac{2n}{\e}\right)^{2n}\dsfrac{1}{\sqrt{2n}} [1+\mathrm{O}(1/4n)]
$$
and
$$
\alpha_n-\tilde\alpha_n=\ln\left[\left(\dsfrac{2n}{\e}\right)^{2n}\dsfrac{1}{\sqrt{2n}}
\dsfrac{n\sqrt{\pi}\, 2^{3/2}}{(2n)! }[1+\mathrm{O}(1/4n)]\right]
$$
Now let us use the well-known formula \cite[p.585, formulae 26]{Lavrentiev Sabat}
$$
(2n)!=\Gamma(2n+1)=\sqrt{4\pi n}\left(\dsfrac{2n}{\e}\right)^{2n}[1+\mathrm{O}(1/2n)]
$$
Then
$$
\alpha_n-\tilde\alpha_n=\ln [1+\mathrm{O}(1/2n)]=\mathrm{O}(1/2n).
$$
Hence we have the following estimate of smallness of the deviation $\beta_n(t)$:
$$
\beta_n(t)=\alpha_n-\tilde\alpha_n +\mathrm{O}\left(\frac{\ln^2\ln t}{\ln t}\right)=
\mathrm{O}\left(\frac{1}{2n}\right)+\mathrm{O}\left(\frac{\ln^2\ln t}{\ln t}\right).
$$
\end{proof}

It means that the  phases  $\alpha_n$ and $\tilde\alpha_n$ and   the  corresponding solitons coincide in the large $n$ and $t$ limit. It also means that a  parametrix at   the  point $\i c$  must be described with the help of the asymptotic solitons (\ref{qass}). Moreover, the solution of this parametrix problem gives the same asymptotic solitons, which were obtained many years ago in \cite{KK} via the Marchenko integral equations of the inverse scattering. For a large $n$ these asymptotic solitons are close to the corresponding  solitons generated by   the   elliptic wave $q_{ell}(x,t)$. Thus the  front  part of the elliptic wave must be eliminated and, instead, this part has to  be the asymptotic solitons generated by  the parametrix at  the  point $\i c$.

Finally, the main term of the asymptotics ($ N,t\gg 1$) of the solution of the nonlinear problem (\ref{mkdv})-(\ref{ic}) is presented as follows:
    $$
    q(x,t)\sim\left\{\begin{array}{ccc}
    c,& x\le -6c^2t;\\ q_{ell}(x,t), & -6c^2t\le x<4c^2t-\dsfrac{2N-3/2}{2c}\ln t;\\
    q_{as}(x,t), & 4c^2t-\dsfrac{2N+1/2}{2c}\ln t <x\le 4c^2t;\\
    0, & x>  4c^2t,
    \end{array}\right.
    $$
where $q_{ell}(x,t)$ and $q_{as}(x,t)$ are defied in (\ref{qmodsumofasymptoticsolitons}) and  (\ref{qass}) respectively.
 
\section{Appendix 1}

\begin{proof}\textbf{of lemma \ref{lemmaF}.\\}
We have that
$$F(\tau^*,z)=\sqrt{c^2-h^2(\tau^*)}\exp\(\dsfrac{-\tau^*}{8}+\dsfrac{\tau^*(z+1)}{4}\)
\dsfrac{\Theta\(\dsfrac{\tau^*(z+1)}{2}|\tau^*\)}
{\Theta\(\dsfrac{\tau^*z}{2}|\tau^*\)}.$$
\\In appendix 2 we prove (see formula (\ref{c2-d2})), that $\sqrt{c^2-h^2(\tau^*)}\exp\(\dsfrac{-\tau^*}{8}\)=4c\(1+\mathrm{O}
\(\e^{\frac{\tau^*}{4}}\)\).$
\\Taking into account the definition of  $\Theta$ function (see formula $(\ref{Delta})$), we can easily verify that for $\tau^*\rightarrow-\infty$
the next formula is true uniformly for $-\dsfrac{1}{2}\leq
z\leq\dsfrac{5}{2}$:
\[F(\tau^*,z)=4c\(1+\mathrm{O}\(\e^{\frac{\tau^*}{4}}\)\)
\e^{\frac{\tau^*(z+1)}{4}}\dsfrac{\mathrm{O}\(\e^{\tau^*(1-z)}\)+\e^{\frac
{-\tau^*z}{2}}+\mathrm{O}\(1\)}
{\mathrm{O}\(\e^{\tau^*(2-z)}\)+\e^{\frac{\tau^*(1-z)}{2}}+1+\mathrm{O}\(
\e^{\frac{\tau^*(1+z)}{2}}\)}.\] Simple calculations lead us to
the next formula:
\[F(\tau^*,z)=\dsfrac{4c\(1+\mathrm{O}\(\e^{\frac{\tau^*}{4}}\)\)
\(1+\mathrm{O}\(\e^{\frac{\tau^*z}{2}}+\e^{\frac{\tau^*(2-z)}{2}}\)\)}
{\e^{\frac{\tau^*(1-z)}{4}}+\e^{\frac{-\tau^*(1-z)}{4}}+\mathrm{O}\(
\e^{\frac{\tau^*(7-3z)}{4}}+\e^{\frac{\tau^*(1+3z)}{4}}\)},
\]
and then
\[F(\tau^*,z)=\dsfrac{4c\(1+\mathrm{O}\(\e^{\frac{\tau^*}{4}}\)\)
\(1+\mathrm{O}\(\e^{\frac{\tau^*z}{2}}+\e^{\frac{\tau^*(2-z)}{2}}\)\)}
{2\cosh\frac{\tau^*(1-z)}{4}\(1+ \mathrm{O}\(
\frac{\e^{\frac{\tau^*(7-3z)}{4}}+\e^{\frac{\tau^*(1+3z)}{4}}}
{\cosh\frac{\tau^*(1-z)}{4}}\)\)}.
\]
Now let us note that
\[\mathrm{O}\(
\frac{\e^{\frac{\tau^*(7-3z)}{4}}+\e^{\frac{\tau^*(1+3z)}{4}}}
{\cosh\frac{\tau^*(1-z)}{4}}\)=\mathrm{O}\(\e^{\frac{\tau^*(3-z)}{2}}+
\e^{\frac{\tau^*(1+z)}{2}}\).\]

Then
\[F(\tau^*,z)=\dsfrac{2c}{\cosh\frac{\tau^*(1-z)}{4}}
\(1+\mathrm{O}\(\e^{\frac{\tau^*}{4}}+
\e^{\frac{\tau^*z}{2}}+\e^{\frac{\tau^*(2-z)}{2}}+
\e^{\frac{\tau^*(3-z)}{4}}+ \e^{\frac{\tau^*(1+z)}{2}}\)\),
\]
and
\[F(\tau^*,z)=\dsfrac{2c}{\cosh\frac{\tau^*(1-z)}{4}}
\(1+\mathrm{O}\(\e^{\frac{\tau^*}{4}}+
\e^{\frac{\tau^*z}{2}}+\e^{\frac{\tau^*(2-z)}{2}} \)\).
\]
Therefore
\[F(\tau^*,z)=\dsfrac{2c}{\cosh\frac{\tau^*(1-z)}{4}}
+\mathrm{O}\(\dsfrac{\e^{\frac{\tau^*}{4}}+
\e^{\frac{\tau^*z}{2}}+\e^{\frac{\tau^*(2-z)}{2}}
}{\cosh\frac{\tau^*(1-z)}{4}}\).
\]
The argument of the $\mathrm{O}$-estimate in the last formula may
be change with:
\[\begin{array}{cc}
    \e^{\frac{\tau^*(1+z)}{4}}&,\textrm{ for }-\dsfrac{1}{2}\leq z\leq\dsfrac{1}{2},\\\\
    \e^{\frac{\tau^*(2-z)}{4}}&,\textrm{ for }\dsfrac{1}{2}\leq z\leq1,\\\\
    \e^{\frac{\tau^*z}{4}}&,\textrm{ for }1\leq z\leq\dsfrac{3}{2},\\\\
    \e^{\frac{\tau^*(3-z)}{4}}&,\textrm{ for }\dsfrac{3}{2}\leq z\leq\dsfrac{5}{2}.\\\\
  \end{array}
\]
Therefore
\[F(\tau^*,z)=\dsfrac{2c}{\cosh\frac{\tau^*(1-z)}{4}}
+\mathrm{O}\(\e^{\frac{\tau^*}{4}}\)\quad \textrm{ for }0\leq
z\leq2,
\]
and
\[F(\tau^*,z)=\mathrm{O}\(\e^{\frac{\tau^*}{8}}\)\quad \textrm{ for }\dsfrac{-1}{2}\leq z\leq\dsfrac{1}{2}.
\]
Also we note that $\dsfrac{2c}{\cosh\frac{\tau^*(1-z)}{4}}=\mathrm{O}\(\e^{\frac{\tau^*}{4}}\)\quad \textrm{ for } z\in\mathbb{R}\backslash(0,2).$\\
Besides this, due to property ($\ref{Thetashift}$) of
$\Theta$-function, we obtain that $F(\tau^*,z)=F(\tau^*,z-2n)$ for
any integer n. Therefore

\[F(\tau^*,z)=\dsfrac{2c}{\cosh\frac{\tau^*(-1-z+2n)}{4}}
+\mathrm{O}\(\e^{\frac{\tau^*}{4}}\)\quad \textrm{ for }2n-2\leq
z\leq2n,
\]
and
\[\dsfrac{2c}{\cosh\frac{\tau^*(-1-z+2n)}{4}}=\mathrm{O}\(\e^{\frac{\tau^*}{4}}\)\quad \textrm{ for } z\in\mathbb{R}\backslash(2n-2,2n).\]
So, we obtain that for any integer $N\geq1$
\[F(\tau^*,z)=\sum_{n=1}^{N}\dsfrac{2c}{\cosh\frac{\tau^*(-1-z+2n)}{4}}+\mathrm{O}\(\e^{\frac{\tau^*}{4}}\),\quad
0\leq z\leq2N.\]
\[F(\tau^*,z)=\mathrm{O}\(\e^{\frac{\tau^*}{8}}\),\quad
\dsfrac{-1}{2}\leq z\leq\dsfrac{1}{2}.\]
\end{proof}

\begin{proof}\textbf{of lemma \ref{lemmaz}.\\}
Recall that
$$z(t,f^{-1}(d(\eta)))=\dsfrac{tB(d(\eta))+\Delta(d(\eta))}{\pi}=8c^3t\eta-\dsfrac{1}{2}+
\dsfrac{2\log\dsfrac{2}{|h_1|}}{\log\dsfrac{1}{\eta}}+\mathrm{O}\(t\eta^2\log\dsfrac{1}{\eta}+
\dsfrac{1}{\log^2\eta}\).$$ From this formula we can deduce that
for any positive $t$ and sufficiently small $\eta$ $:
z\geq\dsfrac{-1}{2}$. Further we will use this remark.
\\By substituting in the last formula for $z$ expression (\ref{etainv}) for $\eta$ in $v$ we obtain the equation ($\ref{z}$).
\end{proof}

\begin{proof}\textbf{of lemma \ref{lemmaintervals}.\\}
Let $\gamma$ and $\delta$ be a some positive numbers such that
$\delta<\gamma$ and let us consider interval
\begin{equation}{\label{xtinterval}}4c^2t-\gamma\log t<x<4c^2 t-\delta\log t
\end{equation}
or
$$1-\dsfrac{\gamma\log t}{4c^2t}<\dsfrac{x}{4c^2t}<1-\dsfrac{\delta\log t}{4c^2t}.$$
We remember that $v=1-\dsfrac{x}{4c^2t}$ and then the last formula
can be rewrite in the next form:
\begin{equation}{\label{v}}\dsfrac{\delta\log t}{4c^2t}<v<\dsfrac{\gamma\log t}{4c^2t}
\end{equation}
We see that for $x,t$, which lay in the interval
$$4c^2t-\gamma\log t<x<4c^2 t-\delta\log t,$$ if
$t\rightarrow\infty$ then $v\rightarrow0$. Further we will use
this fact.
\\
For $t$ and $v$ such that
\begin{equation}{\label{vtinterval}}\dsfrac{\delta\log t}{4c^2t}<v<\dsfrac{\gamma\log t}{4c^2t}\end{equation}
also
\[\dsfrac{1}{\log t}\cdot\dsfrac{1}{1-\frac{\log\log t-\log\frac{4c^2}{\delta}}{\log t}}<\dsfrac{1}{\log\dsfrac{1}{v}}<\dsfrac{1}{\log t}\cdot\dsfrac{1}{1-\frac{\log\log t-\log\frac{4c^2}{\gamma}}{\log t}},\]
and then
\[\dsfrac{2c\delta}{1-\frac{\log\log t-\log\frac{4c^2}{\delta}}{\log t}}<\dsfrac{8c^3tv}{\log\frac{1}{v}}<\dsfrac{2c\gamma}{1-\frac{\log\log t-\log\frac{4c^2}{\gamma}}{\log t}},\]
and
\[\dsfrac{2c\delta}{1-\frac{\log\log t-\log\frac{4c^2}{\delta}}{\log t}}
-\dsfrac{8c^3tv}{\log\frac{1}{v}}+z<z<\dsfrac{2c\gamma}{1-\frac{\log\log
t-\log\frac{4c^2}{\gamma}}{\log t}}
-\dsfrac{8c^3tv}{\log\frac{1}{v}}+z.\] And then, due to formula
($\ref{z}$), we obtain that for any positive $\varepsilon$ there
exists $T=T(\varepsilon, \gamma, \delta)$ such that for any $t>T$,
for any $v$ which satisfy (\ref{v}), and uniformly for
$\delta,\gamma$ from the interval $\dsfrac{1}{4c}<\delta,
\gamma<\dsfrac{2N+1}{2c}$ the following inequality holds true:
\[2c\delta-\dsfrac{1}{2}-\varepsilon\leq z\leq2c\gamma-\dsfrac{1}{2}+\varepsilon\]
\textbf{A. }Now let us take
$\delta=\dsfrac{\varepsilon+\frac{1}{2}}{2c}$,
$\gamma=\dsfrac{2N+\frac{1}{2}-\varepsilon}{2c}$.
\\If we take into account that $(\ref{v})$ is equivalent to $(\ref{xtinterval})$,
then we obtain that for any positive $\varepsilon$ there exists
$T=T(\varepsilon,N)$ so that for any $t>T$ and any $x$ such that
$$4c^2t-\dsfrac{2N+\frac{1}{2}-\varepsilon}{2c}\log t\leq
x\leq4c^2t-\dsfrac{\frac{1}{2}+\varepsilon}{2c}\log t$$ the
following is true:
\[0\leq z\leq2N.
\]
\textbf{B. }If we take now $\gamma=\dsfrac{1-\varepsilon}{2c}$ and
take into account that always for big $t:$ $z\geq\frac{-1}{2}$,
then we obtain that for any positive $\varepsilon$ there exists
$T=T(\varepsilon)$ so that for any $t>T$ and any $x$ such that
$$4c^2t-\dsfrac{1-\varepsilon}{2c}\log t< x\leq4c^2t$$ the
following is true:
\[\dsfrac{-1}{2}\leq z\leq\dsfrac{1}{2}.
\]
\end{proof}

\begin{proof} \textbf{of lemma \ref{lemmaarg}.\\}
Let us express $\dsfrac{\tau^*(z-2n+1)}{4}$ in terms of $v,t$ and
therefore $x,t$. We know that
\[\hskip-2cm\tau^*=-4\log\dsfrac{8}{\eta}+\mathrm{O}(\eta)=4\log v\(1-\dsfrac{\log\log\frac{1}{v}+3\log2}{\log v}+\dsfrac{\log\log\frac{1}{v}+\log8\e}{\log^2v}+
\mathrm{O}\(\dsfrac{\log^2\log\frac{1}{v}}{\log^3v}\)\),\]
\[z=\dsfrac{tB+\Delta}{\pi},\]
\[\frac{1}{\pi}B=8c^3\eta+\mathrm{O}\(\eta^2\log\eta\)=
\dsfrac{8c^3v}{\log\frac{1}{v}}\(1+\dsfrac{\log\log\frac{1}{v}+\log8\e}{\log
v}\)+\mathrm{O}\(\dsfrac{v\log^2\log\frac{1}{v}}{\log^3v}\),\]
\[\frac{1}{\pi}\Delta=\dsfrac{-1}{2}-\dsfrac{2\log2}{\log v}+\mathrm{O}\(\dsfrac{\log\log\frac{1}{v}}{\log^2v}\).\]
And then for $t,v$, which satisfy $(\ref{vtinterval})$, in
forward, but careful calculations we obtain that
\[\dsfrac{\tau^*(z-2n+1)}{4}=-8c^3tv+\(2n-\frac{1}{2}\)\log t+
\(2n-\frac{1}{2}\)\log\dsfrac{\log\frac{1}{v}}{vt}+\dsfrac{8c^3tv}{\log
\frac{1}{v}}+\]\[+\(6n-\dsfrac{7}{2}\)\log2+\mathrm{O}\(\dsfrac{\log^2\log\frac{1}{v}}{\log
v}\).\] And by recalling that $v=1-\frac{x}{4c^2t}$ is a function
of $x,t$, we obtain that
\[\dsfrac{\tau^*(z-2n+1)}{4}=2c(x-4c^2t)+\(2n-\frac{1}{2}\)\log t-\alpha_n(x,t),\]
where
\[\alpha_n(x,t)=
-\(2n-\frac{1}{2}\)\log\dsfrac{\log\frac{1}{v}}{vt}-\dsfrac{8c^3tv}{\log
\frac{1}{v}}-\(6n-\dsfrac{7}{2}\)\log2+\mathrm{O}\(\dsfrac{\log^2\log\frac{1}{v}}{\log
v}\).\] While proving lemma \ref{lemmaintervals} we also have
proved that $\alpha_n(x,t)$ is bounded for $4c^2t-\gamma\log
t<x<4c^2t-\delta\log t$, where $0<\delta<\gamma$ are arbitrary
numbers.
\end{proof}
\section{Appendix 2}
\subsection{$b$-periods $\tau$ and $\tau^*$.}
\begin{lem}{\label{h}} Let
$\tau^*(d)=\dsfrac{4\pi^2}{\tau(d)}$, where $\tau(.)$ is defined
by formula $(\ref{tau})$. Then

\begin{enumerate}
    \item for $d\in(0,c): \tau^*(d)<0$;
    \item $\tau^*(.)$ is decreasing function on the interval
    $d\in(0,c)$;
    \item $\tau^*(+0)=0,\qquad \tau^*(c-0)=-\infty$.
\end{enumerate}
Therefore there exists a function
$h(.):(-\infty,0)\rightarrow(0,c)$, which is the inverse of
$\tau^*(.)$: $\quad h(\tau^*(d))=d$.
\end{lem}

\begin{proof}
Indeed,
\[\tau^*(d)=4\pi\i\(\ds\int\limits_0^{\i d}\dsfrac{\d k}{\w(k,d)}\)
\(\ds\int\limits_{\i d}^{\i c}\dsfrac{\d k}{\w_+(k,d)}\)^{-1}.\]
From the definition of the $\w(k,d)$ we see, that for $k\in(\i
d,\i c)$: $\w_+(k,d)=\i|\w_+(k,d)|$. We make change of variables
$k=\i y$ in the last integrals and get
\[\tau^*(d)=-4\pi\dsfrac{\mathrm{I}_1(d)}{\mathrm{I}_0(d)},\]
where \[\mathrm{I}_1(d)=\ds\int\limits_0^{d}\dsfrac{\d
y}{\sqrt{(c^2-y^2)(d^2-y^2)}},\] and
\[\mathrm{I}_0(d)=\ds\int\limits_{d}^{c}\dsfrac{\d
y}{\sqrt{(c^2-y^2)(y^2-d^2)}}.\] Let make change of variable
$y=sd$ in $\mathrm{I}_1(d).$ Then we see that $\mathrm{I}_1(d)$ is
an increasing function on $d$. Indeed,
\[\mathrm{I}_1(d)=\dsfrac{1}{c}\int\limits_0^1\dsfrac{\d s}{\sqrt{(1-s^2)\(1-\frac{d^2}{c^2}s^2\)}}.\]
Now we make change of variable $y=d+(c-d)s$ in the integral
$\mathrm{I}_0(d).$  Then we see that $\mathrm{I}_0(d)$ is
decreasing function on $d$. Indeed,
\[\mathrm{I}_0(d)=\int\limits_0^1\dsfrac{\d s}{\sqrt{(c(1+s)+d(1-s))(cs+d(2-s))}\sqrt{s(1-s)}}.\]
Then the statements of lemma are evident.
\end{proof}

\subsection{Proof of lemma \ref{lemdmu}.}{\label{dmuxi}}
\begin{proof}
To show that this system has a unique solution let us  define the
function
\begin{equation}\label{Fmud}\nonumber
F(\mu,d)=\int\limits_0^1\left(\mu^2-\lambda^2d^2\right)\sqrt{
 {\frac{1-\lambda^2}{c^2-\displaystyle{\lambda^2d^2}}}}\d\lambda.
\end{equation}
It is easy to see that there  exists a function $\mu=\mu(d)$ such
that $F(\mu(d),d)\equiv0$ and $0<\mu(d)<d$. Moreover, $\mu(d)$ is
strictly increasing in $d\in[0,c]$. Indeed, one can check that
$F(\mu,d)$ is strictly increasing in $\mu$ and is strictly
decreasing in $d$ as $0<\mu<d<c$. Now if $0<d_1<d_2<c$ then
$F(\mu(d_1),d_1)=0=F(\mu(d_2),d_2)<F(\mu(d_2),d_1)$, that is
$F(\mu(d_1),d_1)<F(\mu(d_2),d_1)$ and, hence, $\mu(d_1)<\mu(d_2)$.
Furthermore $\mu(d)$ is continuous function that follows from the
representation:
\begin{equation}\label{mud}
\mu^2(d)= \int\limits_0^1 \lambda^2d^2\sqrt
 {\frac{1-\lambda^2}{c^2-\lambda^2d^2}}d\lambda\Bigg/
\int\limits_0^1  \sqrt
 {\frac{1-\lambda^2}{c^2-\lambda^2d^2}}d\lambda,
\end{equation}
which is equivalent to equality (\ref{dmu1}).
 Equation (\ref{mud}) yields that $\mu(0)=0$ and
$\mu(c)=\displaystyle\frac{c}{\sqrt{3}}.$  Hence
$\mu^2(d)+\displaystyle\frac{d^2}{2}$ varies over the segment
$\left[0,\displaystyle\frac{5c^2}{6}\right]$ as $d$ varies over
the segment $[0,c]$. So for any
$\xi\in\left(-\displaystyle\frac{c^2}{2},\displaystyle\frac{c^2}{3}
\right)$ there exists a unique $d\in(0,c)$ such that (\ref{dmu1})
and (\ref{dmu2}) are fulfilled. Let us denote this function by
$$f:\left[\dsfrac{-c^2}{2},\dsfrac{c^2}{3}\right]\rightarrow[0,c],\qquad\qquad
f(\xi)=d.$$
Moreover, $f$ is increasing function. And therefore there exists an inverse function $f^{-1}:\left[\dsfrac{-c^2}{2},\dsfrac{c^2}{3}\right]\rightarrow[0,c]$.\\
We can now rewrite equality (\ref{dmu2}) as follows:
\[\dsfrac{c^2}{2}+\xi=\mu^2(f(\xi))+\dsfrac{f^2(\xi)}{2}.\]
The last equality and the fact that $\mu(.)$ is increasing imply
that $f(.)$ is continuous function. So, the system (\ref{dmu1}),
(\ref{dmu2}) has a unique solution $\mu=\mu(f(\xi))$, $d=f(\xi)$,
and $\mu(f(\xi)),f(\xi)$ are continuous and strictly increasing
functions.
\end{proof}

\subsection{Asymptotic expansions of $\Delta$.}
In this section we don't always write dependence of some functions
on their argument to simplify the text.
\\
Let us transform expression for $\Delta(d)$ $(\ref{Delta})$:
\begin{equation}{\label{DeltaI2I1}}\Delta(d)=\dsfrac
{\int\limits_{\i d}^{\i c}\dsfrac{\log\left(a_+(k)a_-(k)\right)\d
k}{\w_+(k,d)} } {\i\int\limits_{0}^{\i d}\dsfrac{\d k}{\w(k,d)}}=
-\dsfrac {\int\limits_{\i d}^{\i
c}\dsfrac{\log\left(a_+(k)a_-(k)\right)\d k}{\w_+(k,d)} }
{\int\limits_{0}^{d}\dsfrac{\d y}{\w(\i
y,d)}}=:-\dsfrac{\mathrm{I_2(d)}}{\quad\mathrm{I_1(d)}}
\end{equation}

\subsubsection{Expansion of $\mathrm{I_2}$ in $\eta$.}

First we consider $\mathrm{I}_2(d)$. As for $k\in(\i d,\i c)$
$\w_+(k,d)=\i|\w(k,d)|$, then
\begin{eqnarray}\nonumber
&\mathrm{I}_2(d)=\left\vert k=\i y\right\vert=
\int\limits_d^c\dsfrac{\log\left(a_+(\i y)a_-(\i y)\right)\d
y}{\sqrt{(c^2-y^2)(y^2-d^2)}}= |y=d+(c-d)s|=\\\nonumber &=
(c-d)\int\limits_0^1\dsfrac{\log\left(a_+(\i(d+(c-d)s))a_-(\i(d+(c-d)s))\right)\d
s} {\sqrt{(c-d)(1-s)(c+d+(c-d)s)(c-d)s(2d+(c-d)s)}}=
\\\nonumber&=
\int\limits_0^1\dsfrac{\log\left(a_+(\i(d+(c-d)s))a_-(\i(d+(c-d)s))\right)\d
s} {\sqrt{s(1-s)}\sqrt{(c+d+(c-d)s)(2d+(c-d)s)}},
\end{eqnarray}
and
\begin{eqnarray}
\nonumber&\I_2(d(\eta))=\int\limits_0^1\dsfrac{\log\left(a_+(\i(c(1-\eta)+c\eta
s))a_- (\i(c(1-\eta)+c\eta s))\right)\d s}
{\sqrt{s(1-s)}\sqrt{(c(2-\eta)+c\eta s)(2c(1-\eta)+c\eta s)}}=
\\\nonumber&=\dsfrac{1}{c}\int\limits_0^1\dsfrac{\log\left(a_+(\i(c(1-\eta)+c\eta s))a_-
(\i(c(1-\eta)+c\eta s))\right)\d s}
{\sqrt{s(1-s)}\sqrt{(2-\eta+\eta s)(2-2\eta+\eta s)}}.
\end{eqnarray}
Now consider $\log a_+a_-$. Due to (\ref{asingularity}) we have
that for $k\in(\i c,\i d)$

\begin{eqnarray}\nonumber
&a_-(k)=\dsfrac{h_1}{2}\e^{-\pi\i/4}\sqrt[4]{\left|\dsfrac{2\i
c}{k-\i
c}\right|}\(1+\mathrm{O}\(\sqrt{\left|k-\i c\right|}\)\),\\
\nonumber&a_+(k)=\dsfrac{h_1}{2}\e^{+\pi\i/4}\sqrt[4]{\left|\dsfrac{2\i
c}{k-\i c}\right|}\(1+\mathrm{O}\(\sqrt{\left|k-\i c\right|}\)\)
,\\
\nonumber&\hskip-2cm a_-(k)a_+(k)=\dsfrac{h_1^2}{4}\sqrt{\left|\dsfrac{2\i
c}{k-\i c}\right|}\(1+\mathrm{O}\(\sqrt{\left|k-\i c\right|}\)\)
=|k=\i y|= \dsfrac{h_1^2}{4}\sqrt{\dsfrac{2
c}{c-y}}\(1+\mathrm{O}\(\sqrt{c-y}\)\)=\\
&\nonumber\hskip-2cm= |y=d+(c-d)s|= \dsfrac{h_1^2}{4}\sqrt{\dsfrac{2
c}{(c-d)(1-s)}}\(1+\mathrm{O}\(\sqrt{(c-d)(1-s)}\)\)
=|d=c(1-\eta)|= \\&\nonumber
\hskip-2cm=\dsfrac{h_1^2}{4}\sqrt{\dsfrac{2}{\eta(1-s)}}\(1+\mathrm{O}\(\sqrt{\eta(1-s)}\)\)
=\dsfrac{h_1^2}{2\sqrt{2}\
\eta(1-s)}\(1+\mathrm{O}\(\sqrt{\eta(1-s)}\)\).
\end{eqnarray}
Then
\begin{eqnarray}\nonumber
&\log\left(a_-(k)a_+(k)\right)=
-\dsfrac{1}{2}\log\eta-\dsfrac{1}{2}\log(1-s)+\log\dsfrac{h_1^2}{2\sqrt{2}}+\mathrm{O}\(\eta(1-s)\)=\\&=
-\dsfrac{1}{2}\log\eta-\dsfrac{1}{2}\log(1-s)+\log\dsfrac{h_1^2}{2\sqrt{2}}+\mathrm{O}\(\eta\).
\end{eqnarray}
Then the integral I$_2(d(\eta))$ is equal to
\begin{eqnarray}\nonumber
&\hskip-2cm\mathrm{I}_2(d(\eta))=\int\limits_0^1\dsfrac{-\dsfrac{1}{2}\log\eta-
\dsfrac{1}{2}\log(1-s)+\log\dsfrac{h_1^2}{2\sqrt{2}}+\mathrm{O}\(\eta\)} {c\sqrt{s(1-s)}\sqrt{(2-\eta+\eta
s)(2-2\eta+\eta s)}}\d s=
\\&\hskip-2cm=\nonumber
\int\limits_0^1\dsfrac{-\dsfrac{1}{2c}\log\eta\d s}
{\sqrt{s(1-s)}\sqrt{(2-\eta+\eta s)(2-2\eta+\eta s)}}+
\int\limits_0^1\dsfrac{-\dsfrac{1}{2c}\log(1-s)\d s}
{\sqrt{s(1-s)}\sqrt{(2-\eta+\eta s)(2-2\eta+\eta s)}}+
\\\nonumber&\hskip-2cm+
\int\limits_0^1\dsfrac{\log\dsfrac{h_1^2}{2\sqrt{2}}+\mathrm{O}\(\eta\)}
{c\sqrt{s(1-s)}\sqrt{(2-\eta+\eta s)(2-2\eta+\eta s)}}\d s=
\nonumber
\nonumber\\\nonumber&\hskip-2cm=-\dsfrac{1}{2c}\log\eta\int\limits_0^1\dsfrac{\left(\dsfrac{1}{2}
+\ord\eta\right)\right)\d s} {\sqrt{s(1-s)}}+
\\\nonumber&\hskip-2cm+
\int\limits_0^1\dsfrac{-\dsfrac{1}{2c}\log(1-s)\left(
\dsfrac{1}{2}+ \ord\eta\right) \right)\d s} {\sqrt{s(1-s)}}+
\\\nonumber&\hskip-2cm+
\int\limits_0^1\dsfrac{\dsfrac{1}{2}\log\dsfrac{h_1^2}{2\sqrt{2}}+\mathrm{O}\(\eta\)}
{c\sqrt{s(1-s)}}\d s=
\end{eqnarray}

\begin{eqnarray}\nonumber
\nonumber&\hskip-2cm=-\dsfrac{1}{2c}\log\eta\left(\dsfrac{\pi}{2}+\ord\eta\right)\right)
-\dsfrac{1}{2c}\left(-\pi\log2+\ord\eta\right)\right)+
\left(\dsfrac{\pi}{2c}\log\dsfrac{h_1^2}{2\sqrt{2}}+\mathrm{O}\(\eta\)\right)=
\end{eqnarray}
\begin{eqnarray}\nonumber
\hskip-2cm=-\dsfrac{\pi}{4c}\log\eta+\dsfrac{\pi}{2c}\log{\dsfrac{h_1^2}{\sqrt{2}}}+\ord\eta\right).
\end{eqnarray}
So,
\begin{equation}{\label{I2}}
\mathrm{I_2(d(\eta))}=-\dsfrac{\pi}{4c}\log\eta+\dsfrac{\pi}{2c}\log{\dsfrac{h_1^2}{\sqrt{2}}}+\ord\eta\right)=
\dsfrac{\pi}{4c}\log\dsfrac{h_1^4}{2\eta}+\ord\eta\right).
\end{equation}
The expansion of the I$_1$ is more difficult than expansion of the I$_2$ and is based on a $\Theta$-function identity.
\subsubsection{Expansion of $\tau^*(d(\eta))$ in $\eta$.}
We know (by using Poisson summation formula), that
\[\Theta(z|\tau)=\Theta\left(\dsfrac{2\pi\i z}{\tau}|\dsfrac{4\pi^2}{\tau}\right)
\sqrt{\dsfrac{2\pi}{-\tau}}\left(\exp\dsfrac{-z^2}{2\tau}\right) =
\Theta\left(\dsfrac{\tau^*z}{2\pi\i}|\tau^*\right)
\sqrt{\dsfrac{-\tau^*}{2\pi}}\left(\exp\dsfrac{-z^2\tau^*}{8\pi^2}\right),
\]
where $\tau^*=\dsfrac{4\pi^2}{\tau}$ and from \cite{KM} we know
that\[\dsfrac{\Theta(0|\tau(d))}{\Theta(\pi\i|\tau(d))}=\sqrt{\dsfrac{c+d}{c-d}}\]
(see formula (4.34) in $\cite{KM}$).\\
Let us recall that $\eta(d)=1-\dsfrac{d}{c}$. Then, by using the
Poisson summation formula (\ref{PoissonSum}), we get
\[
\sqrt{\dsfrac{2-\eta(d)}{\eta(d)}}=
\sqrt{\dsfrac{c+d}{c-d}}=\dsfrac{\Theta(0|\tau(d))}{\Theta(\pi\i|\tau(d))}
=\dsfrac{\Theta\left(0|\tau^*(d)\right)}
{\Theta\left(\dsfrac{\tau^*(d)}{2}|\tau^*(d)\right)\exp\dsfrac{\tau^*(d)}{8}}.
\]
Now we use the inverse function $h(.)$ of the function $\tau^*$
(see lemma \ref{h}) and rewrite the last formula.
\[
\sqrt{\dsfrac{2-\eta(h(\tau^*))}{\eta(h(\tau^*))}}=
\dsfrac{\Theta\left(0|\tau^*\right)}
{\Theta\left(\dsfrac{\tau^*}{2}|\tau^*\right)\exp\dsfrac{\tau^*}{8}}.
\]
 Since
\[\Theta(z|\tau^*)=\sum\limits_{m=-\infty}^{\infty}
\exp\left\{\dsfrac{1}{2}\tau^*m^2+zm\right\},
\]
then  $\eta(h(\tau^*))=$
\begin{eqnarray}\nonumber
&\hskip-2cm=2\left( \dsfrac{\Theta^2\left(0|\tau^*\right)}
{\Theta^2\left(\dsfrac{\tau^*}{2}|\tau^*\right)\exp\dsfrac{\tau^*}{4}}
+1\right)^{-1}= 2\left( \dsfrac{\left(
\sum\limits_{m=-\infty}^{\infty}\exp\left\{\dsfrac{1}{2}\tau^*m^2\right\}\right)^2}
{\left(
\sum\limits_{m=-\infty}^{\infty}\exp\left\{\dsfrac{1}{2}\tau^*m(m+1)\right\}\right)^2
\exp\dsfrac{\tau^*}{4} } +1\right)^{-1}=\\\nonumber
\\\nonumber&\hskip-2cm=
\left|\zeta=\e^{\dsfrac{\tau^*}{4}}\right|= 2\left( \dsfrac{\left(
1+\ord \e^{\frac{\tau^*}{2}}\right)\right)^2}
{4e^{\frac{\tau^*}{4}}\left( 1+\ord \e^{\tau^*}\right)\right)^2 }
+1\right)^{-1}= 8\e^{\frac{\tau^*}{4}}+\ord
e^{\frac{\tau^*}{2}}\right),\\\nonumber&\hskip-2.5cm\textrm{Now we get }
\\\nonumber&\eta=8\e^{\frac{\tau^*}{4}}+\ord \e^{\frac{\tau^*}{2}}\right),
\\\nonumber&\e^{\frac{\tau^*}{4}}=\dsfrac{\eta}{8}+\ord\eta^2\right),
\\{\label{tau^*}}&\tau^*=\tau^*(d(\eta))=-4\log\dsfrac{8}{\eta}+\ord\eta\right).
\end{eqnarray}
Also we have
\begin{equation}{\label{c2-d2}}\sqrt{c^2-h^2(\tau^*)}\e^{\frac{-\tau^*}{8}}=\sqrt{c^2-c^2(1-\eta(h(\tau^*)))^2}
\e^{\frac{-\tau^*}{8}}=4c\(1+\mathrm{O}\(\e^{\frac{\tau^*}{4}}\)\).
\end{equation}

\subsubsection{Expansion of $\mathrm{I_1}$ in $\eta$.}
In the other way, (see the proof of lemma $\ref{h}$),
\[\tau^*(d)=-4\pi\dsfrac{\mathrm{I}_1(d)}{\mathrm{I}_0(d)},\] and
then
\begin{equation}{\label{I1}}\mathrm{I}_1(d)=\dsfrac{-\tau^*(d)}{4\pi}\mathrm{I}_0(d).
\end{equation}
Now we get the asymptotic expansion of the $\mathrm{I}_0$.
\begin{eqnarray}\nonumber
&\mathrm{I}_0(d)=\int\limits_{d}^{c}\dsfrac{\d
y}{\sqrt{(c^2-y^2)(y^2-d^2)}}= |y=d+(c-d)s|=
\\\nonumber&=(c-d)\int\limits_0^1\dsfrac{\d s}{\sqrt{(c-d)(1-s)(c+d+(c-d)s)(c-d)s(2d+(c-d)s)}}=
\\&=\int\limits_0^1\dsfrac{\d s}{\sqrt{s(1-s)}\sqrt{(c+d+(c-d)s)(2d+(c-d)s)}}.\nonumber
\end{eqnarray}
Let us make change of variables $d=d(\eta)=c(1-\eta)$ in the last
integral. Then
\begin{eqnarray}
\nonumber&\mathrm{I}_0(d(\eta))=\int\limits_0^1\dsfrac{\d
s}{\sqrt{s(1-s)}\sqrt{(c(2-\eta)+c\eta s)(2c(1-\eta)+c\eta s)}}=
\\\nonumber&=\dsfrac{1}{c}\int\limits_0^1\dsfrac{\d s}{\sqrt{s(1-s)}\sqrt{(2-\eta+\eta s)(2-2\eta+\eta s)}}=
\dsfrac{1}{c}\int\limits_0^1\dsfrac{\left(\dsfrac{1}{2}+\ord\eta\right)\right)\d
s}{\sqrt{s(1-s)}}=
\\\nonumber&=\dsfrac{\pi}{2c}+\ord\eta\right).
\end{eqnarray}
And then by virtue of (\ref{I1}) and ($\ref{tau^*}$) we conclude
that
\begin{equation}{\label{I1ineta}}\mathrm{I}_1(d(\eta))=\left(\dsfrac{1}{2c}+\ord\eta\right)\right)
\left(\log\dsfrac{8}{\eta}+\ord\eta\right)\right)
=\dsfrac{1}{2c}\log\dsfrac{8}{\eta}+\ord\eta\log\eta\right).
\end{equation}

\textbf{Remark.} Although in ($\ref{I1ineta}$) we get only the first member of expansion of
\\$\mathrm{I_1}(kc)=\dsfrac{1}{c}\int\limits_0^1\dsfrac{\d x}{\sqrt{(1-k^2x^2)(1-x^2)}}$ as $k\rightarrow1$
(here $k=\dsfrac{d}{c}$), but in this way we can get as much as desired members of this expansion. See also $\cite{PS}$, problem 90.
\\
Finally, by virtue of ($\ref{DeltaI2I1}$), ($\ref{I1ineta}$) and ($\ref{I2}$), we obtain\\ {\it the expansion of $\Delta(d(\eta))$}:\\
\begin{eqnarray}\nonumber
&\Delta(d(\eta))=
\dsfrac{\dsfrac{\pi}{4c}\log\frac{2\eta}{h_1^4}+\ord\eta\right)}
{\dsfrac{1}{2c}\log\dsfrac{8}{\eta}+ \ord\eta\log\eta\right)}=
\dsfrac{-\pi}{2}\dsfrac{\log\dsfrac{1}{\eta}-\log\frac{2}{h_1^4}+\ord\eta\right)}
{\log\dsfrac{1}{\eta}+
3\log2+\ord\eta\log\eta\right)}=\\\\&\nonumber=
\dsfrac{-\pi}{2}\left(1-\dsfrac{4\log\dsfrac{2}{|h_1|}}{\log\dsfrac{1}{\eta}}+
\ord\dsfrac{1}{\log^2\eta}\right)\right).
\end{eqnarray}
and
\begin{equation}{\label{Deltaineta}}
\dsfrac{1}{\pi}\Delta(d(\eta))=
\dsfrac{-1}{2}\left(1-\dsfrac{4\log\dsfrac{2}{|h_1|}}{\log\dsfrac{1}{\eta}}+
\ord\dsfrac{1}{\log^2\eta}\right)\right).
\end{equation}
\subsection{Asymptotic expansions of $\mu$.}
In the following three paragraphes we get the expansion of $\mu$.
As we know from ($\ref{mud}$),
\begin{equation}\nonumber
\mu^2(d)=\dsfrac{\int\limits_0^d\dsfrac{y^2\sqrt{d^2-y^2}}{\sqrt{c^2-y^2}}\d
y} {\int\limits_0^d\dsfrac{\sqrt{d^2-y^2}}{\sqrt{c^2-y^2}}\d y}.
\end{equation}
And then
\begin{equation}{\label{muinI3I4}}
\hskip-2cm \mu^2(d)=
\dsfrac{d^2\int\limits_0^d\dsfrac{\sqrt{d^2-y^2}}{\sqrt{c^2-y^2}}\d
y} {\int\limits_0^d\dsfrac{\sqrt{d^2-y^2}}{\sqrt{c^2-y^2}}\d y}-
\dsfrac{\int\limits_0^d\dsfrac{(d^2-y^2)^{3/2}}{\sqrt{c^2-y^2}}\d
y} {\int\limits_0^d\dsfrac{\sqrt{d^2-y^2}}{\sqrt{c^2-y^2}}\d y}=
d^2-\dsfrac{\int\limits_0^d\dsfrac{(d^2-y^2)^{3/2}}{\sqrt{c^2-y^2}}\d
y} {\int\limits_0^d\dsfrac{\sqrt{d^2-y^2}}{\sqrt{c^2-y^2}}\d y}=:
d^2-\dsfrac{\I_4(d)}{\I_3(d)}.
\end{equation}

\subsubsection{Expansion of $\mathrm{I_3}$ in $\eta$.} Let us first consider $\I_3$:
\begin{equation}\nonumber
\I_3(d)=\int\limits_0^d\dsfrac{\sqrt{d^2-y^2}}{\sqrt{c^2-y^2}}\d y
\end{equation}
\\or\\
\begin{equation}{\label{I3}}
\I_3(d(\eta))=
\int\limits_0^{c(1-\eta)}\dsfrac{\sqrt{c^2(1-\eta)^2-y^2}}{\sqrt{c^2-y^2}}\d
y,
\end{equation}
where $d(\eta)=c(1-\eta)$.\\
Let us differentiate $I_3(d(.))$ in $\eta$:
\begin{eqnarray}\nonumber
&&\hskip-2cm(\I_3\circ d)^{'}_{\eta}(\eta)=
\int\limits_0^{c(1-\eta)}\dsfrac{-c^2(1-\eta)}
{\sqrt{c^2(1-\eta)^2-y^2}\sqrt{c^2-y^2}}\d y=|\textrm{see
}(\ref{DeltaI2I1})|=
-c^2(1-\eta)\I_1(d(\eta))=\\&&\nonumber\hskip-2cm=|\mathrm{see}
(\ref{I1ineta})|=-c^2(1-\eta)\left(\dsfrac{1}{2c}\log\dsfrac{8}{\eta}+
\ord\eta\log\eta\right)\right)
=\dsfrac{-c}{2}\log\dsfrac{8}{\eta}+\ord\eta\log\eta\right).
\end{eqnarray}
Then
\begin{eqnarray}{\label{I3ineta}} \I_3(d(\eta))&=\I_3(d(0))+\int\limits_0^{\eta}(\I_3\circ d)_{\widetilde{\eta}}^{'}(\widetilde{\eta})\d\widetilde{\eta}\nonumber\\
&=c-\dsfrac{c}{2}\int\limits_0^{\eta}\left(\log\dsfrac{8}{\widetilde{\eta}}+
\ord\widetilde{\eta}\log\widetilde{\eta}\right)\right)\d\widetilde{\eta}=
c-\dsfrac{c}{2}\eta\log\dsfrac{8\e}{\eta}+\ord\eta^2\log\eta\right).
\end{eqnarray}

\subsubsection{Expansion of $\mathrm{I_4}$ in $\eta$.} Now let us consider $\I_4$.
\begin{equation}\nonumber
\I_4(d)=\int\limits_0^d\dsfrac{(d^2-y^2)^{3/2}}{\sqrt{c^2-y^2}}\d
y
\end{equation}
and
\begin{equation}\nonumber
\I_4(d(\eta))=
\int\limits_0^{c(1-\eta)}\dsfrac{(c^2(1-\eta)^2-y^2)^{3/2}}{\sqrt{c^2-y^2}}\d
y
\end{equation}
where $d(\eta)=c(1-\eta)$.\\
Let us differentiate I$_4(d(.))$ in $\eta$:
\begin{eqnarray}\nonumber
&\hskip-2cm\dsfrac{\d\mathrm{I_4\circ}d}{\d\eta}(\eta)=
\int\limits_0^{c(1-\eta)}\dsfrac{-3c^2(1-\eta)\sqrt{c^2(1-\eta)^2-y^2}}
{\sqrt{c^2-y^2}}\d y=|\mathrm{see} (\ref{I3})|=
-3c^2(1-\eta)\I_3(d(\eta))=\\&\nonumber\hskip-2cm=|\mathrm{see}
(\ref{I3ineta})|=-3c^3(1-\eta)\left(1-\dsfrac{1}{2}\eta\log\dsfrac{8\e}{\eta}+\ord\eta^2\log\eta\right)\right)
=-3c^3+\ord\eta\log\eta\right).
\end{eqnarray}
Then
\begin{align}{\label{I4ineta}}\I_4(d(\eta))&=\I_4(d(0))+\int\limits_0^{\eta}(\I\circ d)^{'}_{\widetilde{\eta}}(\widetilde{\eta})\d\widetilde{\eta}\nonumber\\
&=\dsfrac{2}{3}c^3-3c^3\int\limits_0^{\eta}\left(1+
\ord\widetilde{\eta}\log\widetilde{\eta}\right)\right)\d\widetilde{\eta}
=\dsfrac{2}{3}c^3-3c^3\eta+\ord\eta^2\log\eta\right).
\end{align}

\subsubsection{Expansion of $\mu$ in $\eta$.} Finally, by virtue of ($\ref{muinI3I4}$) and the fact that $d=d(\eta)=c(1-\eta)$, we get that
\begin{equation}\nonumber
\mu^2(d)=d^2-\dsfrac{\I_4(d)}{\I_3(d)},
\end{equation}
and
\begin{eqnarray}\nonumber
&\mu^2(d(\eta))=c^2(1-\eta)^2-\dsfrac{\dsfrac{2}{3}c^3-3c^3\eta+\ord\eta^2\log\eta\right)}
{c-\dsfrac{c}{2}\eta\log\dsfrac{8\e}{\eta}+\ord\eta^2\log\eta\right)}=
\\\nonumber&=c^2(1-\eta)^2-\dsfrac{2}{3}c^2\dsfrac{1-\dsfrac{9}{2}\eta+\ord\eta^2\log\eta\right)}
{1-\dsfrac{1}{2}\eta\log\dsfrac{8\e}{\eta}+\ord\eta^2\log\eta\right)}=\\
\nonumber&=c^2(1-\eta)^2-\dsfrac{2}{3}c^2
\left(1-\dsfrac{9}{2}\eta
+\dsfrac{1}{2}\eta\log\dsfrac{8\e}{\eta}+\ord\eta^2\log^2\eta\right)\right)
=\\\nonumber&=c^2(1-\eta)^2+c^2 \left(-\dsfrac{2}{3}+3\eta
-\dsfrac{1}{3}\eta\log\dsfrac{8\e}{\eta}+\ord\eta^2\log^2\eta\right)\right)=
\\\nonumber&=c^2\left(\dsfrac{1}{3}+\eta
-\dsfrac{1}{3}\eta\log\dsfrac{8\e}{\eta}+\ord\eta^2\log^2\eta\right)\right)
=c^2\left(\dsfrac{1}{3}
-\dsfrac{1}{3}\eta\log\dsfrac{8}{\eta\e^2}+\ord\eta^2\log^2\eta\right)\right).
\end{eqnarray}
So,
\begin{equation}{\label{muineta}}
\mu^2(d(\eta))=c^2\left(\dsfrac{1}{3}
-\dsfrac{1}{3}\eta\log\dsfrac{8}{\eta\e^2}+\ord\eta^2\log^2\eta\right)\right)
\end{equation}
and
\begin{equation}{\label{3muinetac2}}\nonumber
\dsfrac{3\mu^2(d(\eta))}{c^2}=\left(1
-\eta\log\dsfrac{8}{\eta\e^2}+\ord\eta^2\log^2\eta\right)\right).
\end{equation}
\subsection{Expansion of $v$ in $\eta$.}
Here we introduce one more small parameter $v$, which is more
close to $\xi$ than $\eta$. As we know from paragraph
($\ref{dmuxi}$),
\[f^{-1}(d)=\mu^2(d)+\dsfrac{d^2}{2}-\dsfrac{c^2}{2}.\]
Let us define new variable\[v=v(\xi)=1-\dsfrac{3\xi}{c^2},\]and remember the definition of the function $d=f(\xi)$ from paragraph $(\ref{dmuxi})$.\\
Then
\[
v(f^{-1}(d))=1-\dsfrac{3}{c^2}\left(\mu^2(d)+\dsfrac{d^2}{2}-
\dsfrac{c^2}{2}\right) \] and
\begin{eqnarray}\nonumber
&\hskip-2cm v(f^{-1}(d(\eta)))=1-\dsfrac{3}{c^2}\left(\mu^2(d(\eta))+\dsfrac{d^2(\eta)}{2}-
\dsfrac{c^2}{2}\right)=1-\left(1
-\eta\log\dsfrac{8}{\eta\e^2}+\ord\eta^2\log^2\eta\right)\right)-
\\\nonumber&\hskip-2cm -\dsfrac{3}{c^2}\left(
\dsfrac{d^2(\eta)}{2}-\dsfrac{c^2}{2}\right)=
\eta\log\dsfrac{8}{\eta\e^2}+\dsfrac{3}{2}
\left(1-\dsfrac{d^2(\eta)}{c^2}\right)+\ord\eta^2\log^2\eta\right)=\\
\nonumber&\hskip-2cm =\eta\log\dsfrac{8}{\eta\e^2}+\dsfrac{3}{2}\eta(2-\eta)+\ord\eta^2\log^2\eta\right)=
\\\nonumber&\hskip-2cm =\eta\log\dsfrac{8\e}{\eta}+\ord\eta^2\log^2\eta\right),
\end{eqnarray}
or
\begin{equation}{\label{vineta}}\dsfrac{v}{8\e}=\dsfrac{\eta}{8\e}\log\dsfrac{8\e}{\eta}
+\ord\eta^2\log^2\eta\right).\end{equation}
\subsection{Expansion of $\eta$ in $v$.}
Let us note that if $x=x(y)$ is invertible function in some
neighborhood of zero and $$x(y)=y+\ord y\right),y\rightarrow0,$$
then $$y(x)=x+\ord x\right),x\rightarrow0.$$ We have
$x=\dsfrac{v}{8\e}$, $y=\dsfrac{\eta}{8\e}\log\dsfrac{8\e}{\eta}$.
Then
\[\dsfrac{\eta}{8\e}\log\dsfrac{\eta}{8\e}=\dsfrac{-v}{8\e}+\ord v^2\right),\quad v\rightarrow0.\]
We have got a Lambert equation
\[w\e^w=z,\quad w<0,z<0,\]
where \begin{equation}{\label{wzineta}}w=\log\dsfrac{\eta}{8\e},\quad z=\dsfrac{-v}{8\e}+\ord v^2\right).\end{equation}
  Following \cite{C&K}, we get that this equation has the solution for $w<<0$, and this inverse has the following expansion:
\[w=-L_1-L_2-\dsfrac{L_2}{L_1}+\ord\dsfrac{L_2^2}{L_1^2}\right),\quad z\rightarrow-0,\]
where
\begin{eqnarray}\nonumber
&L_1=\log\dsfrac{-1}{z},
\\\nonumber&L_2=\log\log\dsfrac{-1}{z}.\end{eqnarray}
Then
\begin{eqnarray}\nonumber
&\e^w=e^{-L_1}\e^{-L_2}\exp\left(-\frac{L_2}{L_1}\right)\exp\left(\ord\dsfrac{L_2^2}{L_1^2}\right)\right),\quad
z\rightarrow-0,
\\\nonumber&\e^w=\dsfrac{-z}{\log{\dsfrac{-1}{z}}}\left(1-\frac{\log\log\dsfrac{-1}{z}}
{\log\dsfrac{-1}{z}}+\ord\dsfrac{\log^2\log\dsfrac{-1}{z}}{\log^2\dsfrac{-1}{z}}\right)\right),\quad
z\rightarrow-0.
\end{eqnarray}
And by virtue of ($\ref{wzineta}$):
\begin{eqnarray}\nonumber
\\&\hskip-2cm\dsfrac{\eta}{8\e}=\dsfrac{\dsfrac{v}{8\e}+\ord v^2\right)}{\log{\dsfrac{1}{\dsfrac{v}{8\e}+\ord v^2\right)}}}\left(1-\frac{\log\log\dsfrac{1}{\dsfrac{v}{8\e}+\ord v^2\right)}}
{\log\dsfrac{1}{\dsfrac{v}{8\e}+\ord
v^2\right)}}+\ord\dsfrac{\log^2\log\dsfrac{1}{\dsfrac{v}{8\e}+\ord
v^2\right)}}{\log^2\dsfrac{1}{\dsfrac{v}{8\e}+\ord
v^2\right)}}\right)\right),\quad z\rightarrow-0,
\nonumber\end{eqnarray}
\begin{eqnarray}
\nonumber\\\nonumber&\eta=\dsfrac{v+\ord
v^2\right)}{\log{\dsfrac{8\e}{v+\ord v^2\right)}}}\left(1-\frac
{\log\log\dsfrac{8\e}{v+\ord v^2\right)}} {\log\dsfrac{8\e}{v+\ord
v^2\right)}}+
\ord\dsfrac{\log^2\log\dsfrac{1}{v}}{\log^2\dsfrac{1}{v}}\right)\right),\quad
v\rightarrow+0.
\end{eqnarray}
As
\begin{eqnarray}\nonumber
\\\nonumber&
\log\log\dsfrac{8\e}{v+\ord v^2\right)}=
\log\log\left(\dsfrac{8\e}{v}\left(1+\ord v\right)\right)\right)=
\log\left(\log\dsfrac{8\e}{v}+\ord v\right)\right)=\\
&=\log\left(\log\dsfrac{8\e}{v}\left(1+\ord\dsfrac{v}{\log
v}\right)\right)\right)=
\log\log\dsfrac{8\e}{v}+\ord\dsfrac{v}{\log v}\right),
\nonumber\end{eqnarray} then
\begin{eqnarray}\nonumber
&\eta=\dsfrac{v}{\log{\dsfrac{8\e}{v}}}\left(1-\frac
{\log\log\dsfrac{8\e}{v}} {\log\dsfrac{8\e}{v}}+
\ord\dsfrac{\log^2\log\dsfrac{1}{v}}{\log^2\dsfrac{1}{v}}\right)\right)
=\\\nonumber&=\dsfrac{v}{\log\dsfrac{1}{v}\left(1+\dsfrac{\log8\e}{\log\dsfrac{1}{v}}\right)}\left(1-\frac
{\log\left(\log\dsfrac{1}{v}\left(1+\dsfrac{\log8\e}{\log\dsfrac{1}{v}}\right)\right)}
{\log\dsfrac{1}{v}\left(1+\dsfrac{\log8\e}{\log\dsfrac{1}{v}}\right)}+
\ord\dsfrac{\log^2\log\dsfrac{1}{v}}{\log^2\dsfrac{1}{v}}\right)\right)
\\\nonumber&
=\dsfrac{v}{\log\dsfrac{1}{v}}\left(1-\frac {\log
8\e+\log\log\dsfrac{1}{v}+\dsfrac{\log8\e}{\log\dsfrac{1}{v}}}
{\log\dsfrac{1}{v}}+
\ord\dsfrac{\log^2\log\dsfrac{1}{v}}{\log^2\dsfrac{1}{v}}\right)\right)
\\\nonumber&
=\dsfrac{v}{\log\dsfrac{1}{v}}\left(1-\frac {\log
8\e+\log\log\dsfrac{1}{v}} {\log\dsfrac{1}{v}}+
\ord\dsfrac{\log^2\log\dsfrac{1}{v}}{\log^2\dsfrac{1}{v}}\right)\right)
,\quad v\rightarrow+0,
\end{eqnarray}
and so we get that
\begin{equation}{\label{etainv}}
\eta=\dsfrac{v}{\log\dsfrac{1}{v}}\left(1-\frac {\log
8\e+\log\log\dsfrac{1}{v}} {\log\dsfrac{1}{v}}+
\ord\dsfrac{\log^2\log\dsfrac{1}{v}}{\log^2\dsfrac{1}{v}}\right)\right)
,\quad v\rightarrow+0.
\end{equation}

\subsection{Asymptotic expansion of $B$.}
As we know, (see $(\ref{Bg})$),
\begin{equation}\nonumber
B(d)=24\ds\int\limits_{\i d}^{\i c}\frac{(k^2+\mu^2(d))(k^2+d^2)\d
k}{\w(k,d)}.
\end{equation}
Now we get the asymptotic expansion of $B$ as $d$ tends to $c$.
\\Let us make change $k=\i y, y\in(d,c)$ in the last integral:
\begin{equation}\nonumber
B(d)=24\ds\int\limits_d^c\frac{(y^2-\mu^2(d))\sqrt{y^2-d^2}\d
y}{\sqrt{c^2-y^2}}= |y=d+(c-d)s|=
\end{equation}\begin{equation}=24(c-d)\ds\int\limits_0^1\frac{((d+(c-d)s)^2-\mu^2(d))\sqrt{(c-d)s(2d+(c-d)s)}\d
s} {\sqrt{(c-d)(1-s)(c+d+(c-d)s)}}=\nonumber
\end{equation}
\begin{equation}=24(c-d)\ds\int\limits_0^1\sqrt\frac{s}{1-s}\frac{((d+(c-d)s)^2-\mu^2(d))\sqrt{(2d+(c-d)s)}}
{\sqrt{(c+d+(c-d)s)}}\d s\nonumber
\end{equation}
Now we make change $d=d(\eta)=c(1-\eta)$ and recall (see
$\ref{muineta}$), that
\begin{equation}\nonumber\mu^2(d(\eta))=\dsfrac{c^2}{3}\left(1+\ord\eta\log\eta\right)\right).
\end{equation}
Then
\begin{eqnarray*}
\hskip-2cm B(d(\eta))&=&24c\eta\ds\int\limits_0^1\sqrt\frac{s}{1-s}\frac{\left(c^2(1-\eta+\eta
s)^2-\dsfrac{1}{3}c^2(1+\ord\eta\log\eta\right))\right)\sqrt{(2c(1-\eta)
+c\eta s)}} {\sqrt{(c(2-\eta)+c\eta s)}}\d
s=\end{eqnarray*}
\begin{eqnarray*}
\hskip-2cm =24c\eta\ds\int\limits_0^1\sqrt\frac{s}{1-s}
\left(\dsfrac{2}{3}c^2\left(1+\ord
\eta\log\eta\right)\right)\right) \left(1+\ord\eta\right)\right)\d
s=\\= 24c\eta\ds\int\limits_0^1\sqrt\frac{s}{1-s}
\left(\dsfrac{2}{3}c^2\left(1+\ord
\eta\log\eta\right)\right)\right)\d s=\end{eqnarray*}
\begin{equation}\label{Bgas}\hskip-2cm =
16c^3\eta\ds\int\limits_0^1\sqrt\frac{s}{1-s}\left(1+\ord\eta\log\eta\right)\right)\d
s==8\pi c^3\eta(1+\ord\eta\log\eta\right)),\quad\eta\rightarrow0.
\end{equation}

\section*{Acknowledgements}
The authors are pleased to acknowledge helpful discussions with E.Ya.Khruslov,
I.~E.~Egorova and D.~G.~Shepelsky relating to the contents of this paper, as well as comments and suggestions given by the Referees.

The research was supported in part by the Akhiezer
foundation and by scholarship of National Academy of Sciences of
Ukraine.

The part of the research was supported by the project "Support of
inter-sectoral mobility and quality enhancement of research teams
at Czech Technical University in Prague" CZ.1.07/2.3.00/30.0034
\bigskip



\section*{References}
\begin{thebibliography}{}

\bibitem{GP} Gurevich A V and  Pitaevskii L P 1973  Decay of Initial Discontinuity in the Korteweg-de Vries Equation
\emph{JETP Letters} \textbf{17/5} 193

\bibitem{BikN1} Bikbaev R F and Novokshenov V Yu 1988 The Korteveg-de Vries Equation with Finite Gap Boundary Conditions and Self-Similar Solutions of Whitham Equations \emph{ Proc. III International Workshop "Nonlinear and Turbulent Processes in Physics"} Kiev \textbf{1} 32-35

\bibitem{BikN2} Bikbaev R F and Novokshenov V Yu 1989 Existence and uniqueness of the solution of the Whitham equation (Russian) \emph{ Asymptotic methods for solving problems in mathematical physics  Akad. Nauk SSSR Ural. Otdel. Bashkir. Nauchn. Tsentr Ufa} 81-95

\bibitem{Bikb1} Bikbaev R F 1989 Structure of a shock wave in the theory of the Korteweg-de Vries equation.  \emph{Phys. Lett. A} \textbf{141/5-6} 289-293

\bibitem{Bikb2} Bikbaev R F and  Sharipov R A 1989 The asymptotic behavior, as $t\to\infty$, of the solution of the Cauchy problem for
the Korteweg-de Vries equation in a class of potentials with
finite-gap behavior as $x\to\pm\infty$.  \emph{ Teoret. Mat.
Fiz.}  \textbf{78/3} 345-356  translation in \emph{Theoret. and
Math. Phys.} \textbf{78/3} 244-252

\bibitem{Bikb3} Bikbaev R F 1989 The Korteweg-de Vries equation with finite-gap boundary conditions and Whitham deformations
of Riemann surfaces (Russian) \emph{ Funktsional. Anal. i Prilozhen.}
\textbf{23/4} 1-10  translation in \emph{Funct. Anal. Appl.}
  \textbf{23/4} 257-266 (1990)

\bibitem{Bikb4} Bikbaev R F 1992 The influence of viscosity on the structure of shock waves in the MKdV model (Russian)
\emph{Zap. Nauchn. Sem. S.-Peterburg. Otdel. Mat. Inst. Steklov. (POMI)  Voprosy Kvant. Teor. Polya Statist. Fiz.} \textbf{11} 37-42
184 translation in \emph{J. Math. Sci.}  \textbf{77}  (1995) \textbf{2} 3042-3045

\bibitem{Bikb5} Bikbaev R F 1995 Complex Whitham deformations in problems with "integrable instability" (Russian)
\emph{Teoret. Mat. Fiz.}  \textbf{104/3} 393-419  translation in
\emph{Theoret. and Math. Phys.} (1996) \textbf{104/3} 1078-1097

\bibitem{Bikb6} Bikbaev R F 1994  Modulational instability stabilization via complex Whitham deformations: nonlinear Schrodinger equation \emph{
Zap. Nauchn. Sem. S.-Peterburg. Otdel. Mat. Inst. Steklov. (POMI)}
\textbf{215}  \emph{Differentsialnaya Geom. Gruppy Li i Mekh. } \textbf{14}
65-76 310 translation in \emph{J. Math. Sci. (New York)}  \textbf{85}  (1997)
\textbf{1} 1596-1604

\bibitem{BK07} Boutet de Monvel A and Kotlyarov V P 2007 Focusing non\-linear Schrodinger equation on the quarter plane with time-periodic boundary condition: a Riemann-Hilbert approach \emph{J. Inst. Math. Jussieu} \textbf{6/4} 579-611

\bibitem{BIK07} Boutet de Monvel A, Its A R and Kotlyarov V P 2007 Long-time asymptotics for the focusing NLS equation with time-periodic
boundary condition. \emph{C. R. Math. Acad. Sci. Paris.} \textbf{
345/11} 615-620

\bibitem{BIK09}  Boutet de Monvel A, Its A R and Kotlyarov V P 2009  Long -time asymptotics for the focusing NLS equation with time -- periodic
boundary condition on the half line. \emph{Comm. Math. Phys.}  \textbf{290/2} 479-522

\bibitem{BKS11}  Boutet de Monvel A, Kotlyarov V P and Shepelsky D G 2011
 Focusing NLS equation: long-time dynamics of step-like initial
data.  \emph{International Mathematics Research Notices} \textbf{7} 1613-1653

\bibitem{BV} Buckingham R and Venakides S 2007 Long-time asymptotics of the non-linear Schrodinger equation shock problem. \emph{Comm. Pure Appl. Math.} \textbf{60/9} 1349-1414

\bibitem{C&K} Corless Robert M,  Gonnet G H, Hare D E G, Jeffrey D J and Knuth D E 1996 On the Lambert W Function \emph{ Advances in Computational Mathematics}
\textbf{5} 329-359

\bibitem{DZ93} Deift P and Zhou X 1993 A steepest descent method for oscillatory Riemann -- Hilbert problems. Asymptotics for the MKdV equation \emph{
 Annals of Mathematics} \textbf{137/2} 295-368

\bibitem{EGKT} Egorova I, Gladka Z, Kotlyarov V and Teschl G 2012  Long-Time Asymptotics for the Korteweg-de Vries Equation with Steplike Initial
Data. \emph{Nonlinearity } \textbf{26/7} 1839-1864

\bibitem{Kh1} Khruslov E Ya 1975 Splitting of an initial step-like perturbation for the KdV equation \emph{ Letters to JETP} \textbf{21/4} 469-472

\bibitem{Kh2} Khruslov E Ya 1976 Asymptotics of the solution of the Cauchy prob\-lem for the Korteweg de Vries equation with initial data of step type. \emph{Matem. Sbornik (New Series)} \textbf{99(141):2} 261-281

\bibitem{KK} Khruslov E Ya and Kotlyarov V P 1989 Asymptotic
solitons of the modified Korteweg-de Vries equation \emph{  Inverse
problems} \textbf{5/6} 1075-1088

\bibitem{KK2} Khruslov E Ya and Kotlyarov V P 1994 Soliton asymptotics of nondecreasing solutions of nonlinear completely integrable evolution equations \emph{    Spectral operator theory and
related topics} Adv. Soviet Math. \textbf{19} Amer. Math. Soc.
Providence, RI 129-180

\bibitem{KK3} Khruslov E Ya and  Kotlyarov V P 2003  Generation of asymptotic solitons in an integrable model of stimulated Raman scattering by periodic boundary data \emph{Mat. Fiz.Anal. Geom.}  \textbf{10/3} 366-384

\bibitem{KM} Kotlyarov Vladimir and Minakov Alexander 2010
Riemann-Hilbert problem to the modified Korteveg de Vries equation:
Long-time dynamics of the steplike initial data
\emph{Journal of Mathematical Physics} \textbf{51} 093506

\bibitem{Lavrentiev Sabat} Lavrent'ev M A and Sabat B V 1951 Metody teorii funkcii kompleksnogo peremennogo. (Russian) Methods of the theory of functions of a complex variable \emph{ Gosudarstv. Izdat. Tehn.-Teor. Lit. Moscow-Leningrad}

\bibitem{M1} Minakov A 2011 Long-time behavior of the solution to the mKdV
equation with step-like initial data \emph{ J. Phys. A: Math. Theor.} \textbf{44} 085206

\bibitem{M2} Minakov A 2011
Asymptotics of Rarefaction Wave Solution
to the mKdV Equation \emph{Journal of Mathematical Physics, Analysis, Geometry} \textbf{ 7/1} 59-86

\bibitem{KM2}  Kotlyarov V and Minakov A 2011
Step-Initial Function to the MKdV Equation:
Hyper-Elliptic Long-Time Asymptotics of the Solution \emph{ Journal of Mathematical Physics, Analysis, Geometry} \textbf{ 8/1} 37-61

\bibitem{MK} Moskovchenko E A and Kotlyarov V P 2006 A new Riemann -- Hilbert problem in a model of stimulated Raman scattering \emph{
J.Phys.A.: Math. Gen.} \textbf{39} 014591

\bibitem{Mos} Moskovchenko E A 2009 Simple periodic boundary data and Rie\-mann -- Hilbert problem for integrable model of the stimulated Raman scattering
    \emph{Journal of mathematical physics,analysis, geometry} \textbf{5/1} 82-103

\bibitem{MK10} Moskovchenko E A and Kotlyarov V P  Periodic boundary data for an integrable model of stimulated Raman scattering: long-time asymptotic behavior
    \emph{Journal of Physics A: Mathematical and Theoretical}\textbf{ 43/5} 055205

\bibitem{Novik} Novokshenov V Yu 2003 Time asymptotics for soliton equations in problems with step initial conditions (Russian) \emph{Sovrem. Mat. Prilozh., Asimptot. Metody Funkts. Anal.} \textbf{ 5} 138-168 translation in \emph{J. Math. Sci. (N. Y.)} \textbf{ 125} \textbf{5} 717-749  (2005)

\bibitem{Novokshenov80} Novoksenov V Ju 1980  Asymptotic behavior as $t\to\infty$ of the solution of the Cauchy problem for a nonlinear
Schrodinger equation (Russian) \emph{ Dokl. Akad. Nauk SSSR}\textbf{ 251/4} 799-802

\bibitem{Novokshenov82}Novokshenov V Yu 1982 Asymptotic Formulae for the Solutions of the System of Nonlinear Schrodinger
Equations \emph{ Uspekhi Matem. Nauk} \textbf{37/2} 215-216

\bibitem{Novokshenov85}Novokshenov V Yu 1982  Asymptotics as $t\to\infty$ of the Solution to a Two-Dimentional Generalisation of the Toda Lattice
\emph{Doklady AN SSSR} \textbf{265/6} 1320-1324 translation in \emph{Soviet Math. Dokl}  \textbf{26/1} 264-268 (1982)

\bibitem{PS} Polya George and Szego Gabor 1978 Problems and theorems in analysis. I. Series, integral calculus, theory of functions
Translated from the German by Dorothee Aeppli \emph{ Reprint of the 1978
English translation. Classics in Mathematics} Springer-Verlag Berlin xx+389 pp. ISBN: 3-540-63640-4



\end{thebibliography}
\end{document}